\documentclass[fullpage,11pt]{article}

\usepackage{fullpage,amsthm}
\usepackage{graphicx} 
\usepackage{array} 

\usepackage{amsmath, amssymb, amsfonts, verbatim}
\usepackage{hyphenat,epsfig,subfigure,multirow}

\usepackage[usenames,dvipsnames]{xcolor}
\usepackage[ruled]{algorithm2e}

\usepackage{tcolorbox}
\usepackage{balance}

\definecolor{DarkRed}{rgb}{0.5,0.1,0.1}
\definecolor{DarkBlue}{rgb}{0.1,0.1,0.5}

\usepackage{hyperref}
\hypersetup{
colorlinks=true,
pdfnewwindow=true,
citecolor=ForestGreen,
linkcolor=DarkRed,
filecolor=DarkRed,
urlcolor=DarkBlue
}

\usepackage{bm}
\usepackage{url}
\usepackage{xspace} 
\usepackage[mathscr]{euscript}

\usepackage[noend]{algpseudocode}
\makeatletter
\def\BState{\State\hskip-\ALG@thistlm}
\makeatother

\usepackage{cite}
\usepackage{enumerate}

\usepackage[margin=1in]{geometry}

\newtheorem{theorem}{Theorem}
\newtheorem{lemma}{Lemma}[section]
\newtheorem{proposition}[lemma]{Proposition}
\newtheorem{corollary}[theorem]{Corollary}
\newtheorem{claim}[lemma]{Claim}
\newtheorem{fact}[lemma]{Fact}

\newtheorem{definition}{Definition}
\newtheorem{problem}{Problem}
\newtheorem{remark}[lemma]{Remark}
\newtheorem*{claim*}{Claim}
\newtheorem*{proposition*}{Proposition}
\newtheorem*{lemma*}{Lemma}
\newtheorem*{problem*}{Problem}

\allowdisplaybreaks

\renewcommand{\qed}{\nobreak \ifvmode \relax \else
      \ifdim\lastskip<1.5em \hskip-\lastskip
      \hskip1.5em plus0em minus0.5em \fi \nobreak
      \vrule height0.75em width0.5em depth0.25em\fi}

\usepackage[T1]{fontenc}
\usepackage[utf8]{inputenc}

\newcommand{\ourinfo}[1]{Department of Computer and Information Science, University of Pennsylvania. Supported in part by National Science
  Foundation grants CCF-1116961, CCF-1552909, CCF-1617851, and IIS-1447470.  \newline\noindent Email: \texttt{#1}.}

\newcommand{\bhhz}{\ensuremath{\bhh^{0}}\xspace}

\newcommand{\toShrink}{-.20cm}
\newcommand{\toShrinkEnu}{-.2cm}

\newcommand{\MSa}[1]{\ensuremath{\textnormal{\textsf{Matching}}_{#1}}}

\newcommand{\distMSa}{\ensuremath{\dist_{\textnormal{\textsf{M}}}}}

\newcommand{\protMSa}{\ensuremath{\Prot_{\ms}}\xspace}
\newcommand{\bProtMSa}{\ensuremath{\bm{\Prot}_{\ms}}\xspace}

\newcommand{\bT}{\bm{\theta}}

\newcommand{\itfacts}[1]{Fact~\ref{fact:it-facts}-(\ref{part:#1})\xspace}

\newcommand{\fistar}{\ensuremath{f_{\istar}}}
\newcommand{\gistar}{\ensuremath{g_{\istar}}}

\newcommand{\bRms}{\ensuremath{\bR_{\textnormal{\textsf{M}}}}}

\newcommand{\sms}{\ensuremath{\textnormal{\textsf{SMS}}}\xspace}
\newcommand{\bhm}{\ensuremath{\textnormal{\textsf{BHM}}}\xspace}
\newcommand{\bhmz}{\ensuremath{\textnormal{\textsf{BHM}}^0}\xspace}

\newcommand{\SMS}[2]{\ensuremath{\sms_{#1,#2}}\xspace}
\newcommand{\DistBHM}{\ensuremath{\dist_{\bhm}}}

\newcommand{\DistBHMY}{\ensuremath{\dist^{\textsf{Y}}_{\bhm}}}
\newcommand{\DistBHMN}{\ensuremath{\dist^{\textsf{N}}_{\bhm}}}

\newcommand{\DistSMS}{\ensuremath{\dist_{\sms}}}
\newcommand{\DistSMSY}{\ensuremath{\dist^{\textsf{Y}}_{\sms}}}
\newcommand{\DistSMSN}{\ensuremath{\dist^{\textsf{N}}_{\sms}}}

\newcommand{\PiSMS}{\ensuremath{\Pi_{\sms}}}
\newcommand{\protSMS}{\ensuremath{\Prot_{\sms}}\xspace}
\newcommand{\protBHM}{\ensuremath{\Prot_{\bhm}}\xspace}

\newcommand{\PiBHM}{\ensuremath{\Pi_{\bhm}}}

\newcommand{\pfYes}[1]{\ensuremath{p^{\textsf{Y}}_{#1}}}
\newcommand{\pfNo}[1]{\ensuremath{p^{\textsf{N}}_{#1}}}

\newcommand{\fYes}{\ensuremath{f_{\Yes}}}
\newcommand{\fNo}{\ensuremath{f_{\No}}}

\newcommand{\Mbhm}{\ensuremath{M^{\textsf{B}}}}
\newcommand{\xbhm}{\ensuremath{x^{\textsf{B}}}}

\newcommand{\bsigma}{\ensuremath{\bm{\sigma}}}

\newcommand{\algline}{
  \rule{0.5\linewidth}{.1pt}\hspace{\fill}%
  \par\nointerlineskip \vspace{.1pt}
}

\newcommand{\MO}[1]{\ensuremath{\textnormal{\textsf{Matching}}_{#1}}}

\newcommand{\mo}{\ensuremath{\textnormal{\textsf{Matching}}}\xspace}
\newcommand{\bhh}{\ensuremath{\textnormal{\textsf{BHH}}}\xspace}
\newcommand{\distMO}{\ensuremath{\dist_{\textnormal{\textsf{M}}}}}

\newcommand{\cmax}{\ensuremath{\textnormal{\textsf{c}}_{\textnormal{rs}}}}
\newcommand{\protMO}{\ensuremath{\Prot_{\mo}}\xspace}
\newcommand{\bProtMO}{\ensuremath{\bm{\Prot}_{\mo}}\xspace}
\newcommand{\protBHH}{\ensuremath{\Prot_{\bhh}}\xspace}
\newcommand{\bProtBHH}{\ensuremath{\bm{\Prot}_{\bhh}}\xspace}

\newcommand{\MS}[1]{\ensuremath{\textnormal{\textsf{Matching}}_{#1}}}

\newcommand{\ms}{\ensuremath{\textnormal{\textsf{Matching}}}\xspace}

\newcommand{\distMS}{\ensuremath{\dist_{\textnormal{\textsf{M}}}}}

\newcommand{\protMS}{\ensuremath{\Prot_{\ms}}\xspace}
\newcommand{\bProtMS}{\ensuremath{\bm{\Prot}_{\ms}}\xspace}

\newcommand{\vstari}{\ensuremath{V^*_i}\xspace}

\newcommand{\Mrs}{\ensuremath{M^{\textsf{RS}}}}
\newcommand{\jstar}{\ensuremath{{j^{\star}}}}

\newcommand{\Grs}{\ensuremath{G^{\textsf{RS}}}}

\newcommand{\RS}[1]{\ensuremath{\textnormal{\textsf{RS}}(#1)}}
\newcommand{\Ibhh}{\ensuremath{I_{\bhh}}}
\newcommand{\Mstar}{\ensuremath{M^{\star}}}
\newcommand{\negative}[1]{\ensuremath{\overline{#1}}}

\newcommand{\FL}{\ensuremath{\mathcal{L}}}
\newcommand{\FS}{\ensuremath{\mathcal{S}}}

\newcommand{\tvd}[2]{\ensuremath{\norm{#1 - #2}_{tvd}}}
\newcommand{\Ot}{\ensuremath{\widetilde{O}}}
\newcommand{\eps}{\ensuremath{\varepsilon}}
\newcommand{\Paren}[1]{\Big(#1\Big)}
\newcommand{\Bracket}[1]{\Big[#1\Big]}
\newcommand{\bracket}[1]{\left[#1\right]}
\newcommand{\paren}[1]{\ensuremath{\left(#1\right)}\xspace}
\newcommand{\card}[1]{\left\vert{#1}\right\vert}

\newcommand{\half}{\ensuremath{{1 \over 2}}}

\newcommand{\norm}[1]{\ensuremath{\|#1\|}}

\newcommand{\floor}[1]{{\left\lfloor{#1}\right\rfloor}}
\newcommand{\prob}[1]{\Pr\paren{#1}}
\newcommand{\expect}[1]{\Exp\bracket{#1}}
\newcommand{\var}[1]{\textnormal{Var}\bracket{#1}}

\newcommand{\set}[1]{\ensuremath{\left\{ #1 \right\}}}
\newcommand{\poly}{\mbox{\rm poly}}
\newcommand{\polylog}{\mbox{\rm  polylog}}

\newcommand{\OPT}{\ensuremath{\mbox{\sc opt}}\xspace}
\newcommand{\opt}{\textnormal{\ensuremath{\mbox{opt}}}\xspace}

\newcommand{\ALG}{\ensuremath{\mbox{\sc alg}}\xspace}

\DeclareMathOperator*{\Exp}{\ensuremath{{\mathbb{E}}}}
\DeclareMathOperator*{\Prob}{\ensuremath{\textnormal{Pr}}}
\renewcommand{\Pr}{\Prob}
\newcommand{\EX}{\Exp}
\newcommand{\Ex}{\Exp}

\newcommand{\etal}{{\it et al.\,}}

\newcommand{\rs}{Ruzsa-Szemer\'{e}di\xspace}

\newcommand{\Union}{\ensuremath{\bigcup}}

\newcommand{\tOPT}{\ensuremath{\widehat{\OPT}}}

\newenvironment{tbox}{\begin{tcolorbox}[
		enlarge top by=5pt,
		enlarge bottom by=5pt,
		 boxsep=0pt,
                  left=4pt,
                  right=4pt,
                  top=10pt,
                  arc=0pt,
                  boxrule=1pt,toprule=1pt,
                  colback=white
                  ]
	}
{\end{tcolorbox}}

\newcommand{\BHM}[1]{\ensuremath{\textnormal{\textsf{BHM}}_{#1}}\xspace}
\newcommand{\BHMZ}[1]{\ensuremath{\textnormal{\textsf{BHM}}^{0}_{#1}}\xspace}

\newcommand{\BHH}[2]{\ensuremath{\textnormal{\textsf{BHH}}_{#1,#2}}\xspace}
\newcommand{\BHHZ}[2]{\ensuremath{\textnormal{\textsf{BHH}}^{0}_{#1,#2}}\xspace}

\newcommand{\Yes}{\ensuremath{\textnormal{\textsf{Yes}}}\xspace}
\newcommand{\No}{\ensuremath{\textnormal{\textsf{No}}}\xspace}

\newcommand{\dist}{\ensuremath{\mathcal{D}}}
\newcommand{\distbhh}{\ensuremath{\mathcal{D}_{\textnormal{\textsf{BHH}}}}}

\newcommand{\Prot}{\ensuremath{\Pi}}
\newcommand{\prot}{\ensuremath{\pi}}

\newcommand{\bProt}{\bm{\Prot}}

\newcommand{\bX}{\bm{X}}
\newcommand{\bY}{\bm{Y}}
\newcommand{\bZ}{\bm{Z}}
\newcommand{\bA}{\bm{A}}
\newcommand{\bP}{\bm{P}}
\newcommand{\bR}{\bm{R}}

\newcommand{\bJ}{\bm{J}}
\newcommand{\bSigma}{\bm{\sigma}}

\newcommand{\bB}{\ensuremath{\bm{B}}}
\newcommand{\bC}{\ensuremath{\bm{C}}}

\newcommand{\supp}[1]{\ensuremath{\textsc{supp}(#1)}}

\newcommand{\ICost}[2]{\ensuremath{\textnormal{\textsf{ICost}}_{#2}(#1)}\xspace}

\newcommand{\ICS}[3]{\ensuremath{\textnormal{\textsf{IC}}_{\textnormal{SMP},#2}^{#3}(#1)}\xspace}
\newcommand{\CCS}[3]{\ensuremath{\textnormal{\textsf{CC}}_{\textnormal{SMP},#2}^{#3}(#1)}\xspace}

\newcommand{\ICO}[3]{\ensuremath{\textnormal{\textsf{IC}}_{1\textnormal{-way},#2}^{#3}(#1)}\xspace}
\newcommand{\CCO}[3]{\ensuremath{\textnormal{\textsf{CC}}_{1\textnormal{-way},#2}^{#3}(#1)}\xspace}

\newcommand{\Matching}{\ensuremath{\textnormal{\textsf{Matching}}}}

\newcommand{\distmatch}{\ensuremath{\mathcal{D}_{\textnormal{MM}}}}

\newcommand{\istar}{\ensuremath{i^{\star}}}

\newcommand{\HM}{\ensuremath{\mathcal{M}}}

\newcommand{\bE}{\ensuremath{\bm{E}}}
\newcommand{\player}[1]{\ensuremath{P^{(#1)}}}
\newcommand{\PS}[1]{\player{#1}}

\newcommand{\textbox}[2]{
{
\begin{tbox}
\textbf{#1}
{#2}
\end{tbox}
}
}

\newcommand{\galpha}{\ensuremath{G^{\alpha}}}
\newcommand{\galphap}{\ensuremath{G^{\beta}}}

\newcommand{\tester}{\ensuremath{\textnormal{\textsf{Tester}}}\xspace}

\newcommand{\sample}[2]{\ensuremath{\textnormal{\textsf{Sample}}_{#2}(#1)}\xspace}

\newcommand{\GS}{\ensuremath{G_{\textnormal{\textsf{smp}}}}}
\newcommand{\VS}{\ensuremath{V_{\textnormal{\textsf{smp}}}}}
\newcommand{\MstarS}{\ensuremath{\Mstar_{\textnormal{\textsf{smp}}}}}

\newcommand{\topt}{\ensuremath{\widetilde{\textnormal{opt}}}}
\renewcommand{\tOPT}{\topt}
\renewcommand{\OPT}{\opt}

\newcommand{\testerg}{\ensuremath{\tester_{\gamma}}\xspace}

\title{On Estimating Maximum Matching Size in Graph Streams}
\author{Sepehr Assadi\thanks{\ourinfo{\{sassadi,sanjeev,yangli2\}@cis.upenn.edu}}  \and Sanjeev Khanna\footnotemark[1] \and Yang Li\footnotemark[1]}

\date{}

\begin{document}
\maketitle

\thispagestyle{empty}
\begin{abstract}
  We study the problem of estimating the maximum matching \emph{size} in graphs whose edges are revealed in a streaming manner. We
  consider both \emph{insertion-only} streams, which only contain edge insertions, and \emph{dynamic} streams that allow both 
  insertions and deletions of the edges, and present new upper and lower bound results for both cases. 
  
On the upper bound front, we show that an $\alpha$-approximate estimate of the matching size can be computed in dynamic streams using $\Ot({n^2 / \alpha^4})$ space, and in insertion-only streams using $\Ot(n/\alpha^2)$-space. These bounds respectively shave off a factor of $\alpha$ from the space necessary to compute an $\alpha$-approximate matching (as opposed to only size), thus proving a non-trivial separation between approximate estimation and approximate computation of matchings in data streams.

On the lower bound front, we prove that any $\alpha$-approximation algorithm for estimating matching size in dynamic graph streams requires $\Omega(\sqrt{n}/\alpha^{2.5})$ bits of space, 
\emph{even} if the underlying graph is both \emph{sparse} and has \emph{arboricity} bounded by $O(\alpha)$. 
We further improve our lower bound to $\Omega(n/\alpha^2)$ in the case of \emph{dense} graphs. These results establish the first non-trivial streaming lower bounds for \emph{super-constant} approximation of 
matching size. 
	
Furthermore, we present the first \emph{super-linear} space lower bound for computing a $(1+\eps)$-approximation of matching size \emph{even} in insertion-only streams. In particular, we prove that a $(1+\eps)$-approximation to matching size requires $\RS{n} \cdot n^{1-O(\eps)}$ space; here, $\RS{n}$ denotes the maximum number of edge-disjoint \emph{induced matchings} of size $\Theta(n)$ in an $n$-vertex graph.
It is a major open problem with far-reaching implications to determine the value of $\RS{n}$, and current results leave open the possibility that $\RS{n}$ may be as large as $n/\log n$.
Moreover, using the best known lower bounds for $\RS{n}$, our result already rules
out any $O(n \cdot \poly(\log{n}/\eps))$-space algorithm for $(1+\eps)$-approximation of matchings. We also show how to avoid the dependency on the parameter $\RS{n}$
in proving lower bound for dynamic streams and present a near-optimal lower bound of $n^{2-O(\eps)}$ for $(1+\eps)$-approximation in this model.  

Using a well-known connection between matching size and \emph{matrix rank}, all our lower bounds also hold for the problem of estimating matrix rank. In particular our results imply a near-optimal $n^{2-O(\eps)}$ bit lower bound for 
$(1+\eps)$-approximation of matrix ranks for dense matrices in dynamic streams, answering an open question of Li and Woodruff (STOC 2016).

\end{abstract}

\clearpage
\setcounter{page}{1}

\section{Introduction}\label{sec:intro}
Recent years have witnessed tremendous progress on solving graph optimization problems in the \emph{streaming} model of computation, formally
introduced in the seminal work of~\cite{AlonMS96}. In this model, a graph is presented as a stream of edge
insertions (\emph{insertion-only streams}) or edge insertions and deletions (\emph{dynamic streams}), and the goal is to solve the given problem with minimum space requirement (see a survey by McGregor~\cite{M14} for a summary). 

One of the most extensively studied problems in the streaming literature is the classical 
problem of finding a \emph{maximum matching}~\cite{Lovasz09}.
Although significant advances have been made on understanding the space needed to compute a maximum matching in the streaming model~\cite{M05,FKMSZ05,EKS09,EpsteinLMS11,GoelKK12,KonradMM12,Zelke12,AGM12,AG13,GO13,Kapralov13,KapralovKS14,CS14,ChitnisCHM15,M14,AhnG15,EsfandiariHLMO15,Konrad15,AssadiKLY16,ChitnisCEHMMV16,BuryS15,McgregorV16}, some important problems remain wide open. In particular, not much is known about the tradeoff between space and approximation for the problem of estimating the {\em size} of a maximum matching in the streaming model. 

In this paper, we obtain new upper and lower bounds for the matching size problem.  Our results show that while the problem of matching size
estimation is provably easier than the problem of finding an approximate matching, the space complexity of the two problems starts to
converge together as the accuracy desired in the computation approaches near-optimality. In particular, we establish the first
super-linear space lower bound (in $n$) for the matching size estimation problem. A well-known connection between matching size and matrix
rank allows us to carry our lower bound results to the problem of estimating rank of a matrix in the streaming model, and we show that
essentially quadratic space is necessary to obtain a near-optimal approximation of matrix rank.  In what follows, we first briefly review
the previous work, and then present our results and techniques in detail.

\subsection{Models and Previous Work}
Two types of streams are generally considered in the literature, namely insertion-only streams and dynamic streams. In insertion-only
streams, edges are only inserted, and in dynamic streams, edges can be both inserted and deleted. 
In the following, we briefly summarize previous results for \emph{single-pass} algorithms (i.e., algorithms that only make one pass over the steam) in both insertion-only streams and dynamic streams.

\paragraph{Insertion-only streams.} It is easy to compute a $2$-approximate matching using $\Ot(n)$ space in insertion-only streams:
simply maintain a \emph{maximal} matching during the stream; here $n$ denotes the number of vertices in the input graph. This can be done similarly
for computing an $\alpha$-approximate matching in $\Ot(n/\alpha)$ space for any $\alpha \geq 2$. 
On the lower bound side, it is shown in~\cite{Kapralov13,GoelKK12} that computing better than a $e/(e-1)$-approximate matching requires $n^{1 + \Omega(1/\log\log n)}$ space. 

For the seemingly easier problem of estimating the maximum matching \emph{size} (the focus of this paper), the result of~\cite{Kapralov13,GoelKK12}
can be modified to show that computing better than a $e/(e - 1)$-approximation for matching size requires 
$n^{\Omega(1/\log\log n)}$ space (see also~\cite{KapralovKS14}).  It was shown later in~\cite{EsfandiariHLMO15} that computing better than a $3/2$-approximation requires
$\Omega(\sqrt{n})$ bits of space. More recently, this lower bound was extended by~\cite{BuryS15} to show that computing a
$(1+\eps)$-estimation requires $n^{1 - O(\eps)}$ space.  On the other hand, the only existing non-trivial algorithm is a folklore that an
$O(\sqrt{n})$-approximation can be obtained in $\polylog(n)$ space even in dynamic streams (for completeness, we provide a self-contained proof of this result in Appendix~\ref{app:folklore}). We note that  other algorithms
that use $o(n)$ space for this problem also exist, but they only work under certain conditions on the input: either the edges are presented in a \emph{random
order}~\cite{KapralovKS14} or the input graph has \emph{bounded arboricity}~\cite{EsfandiariHLMO15,ChitnisCEHMMV16,BuryS15,McgregorV16}.

\paragraph{Dynamic streams.} Space complexity of finding an $\alpha$-approximate matching in dynamic graph streams is essentially resolved: it is shown in~\cite{AssadiKLY16} that $\tilde{\Theta} ({n^2/\alpha^3})$ space is  
\emph{necessary} and in~\cite{AssadiKLY16, ChitnisCEHMMV16}, that it is also \emph{sufficient} (see also~\cite{Konrad15}).
However, the space complexity of estimating the matching size (the focus of this
paper) is far from being settled in this model. For example, it is not even known if $\alpha$-approximating matching size is
strictly easier than finding an $\alpha$-approximate matching (for any $\alpha = o(\sqrt{n})$). Moreover, no better lower bounds are known for estimating matching size in
dynamic streams than the ones in~\cite{EsfandiariHLMO15,BuryS15}, which already hold even for insertion-only streams. 

This state-of-the-art in both insertion-only and dynamic streams highlights the following natural question: \emph{How well can we approximate the maximum 
matching size in a space strictly smaller that what is needed for finding an approximate matching? 
In general, what is the space-approximation tradeoff for estimating the matching size in graph streams?}

Indeed, this question (and its closely related variants) has already been raised in the literature~\cite{EsfandiariHLMO15,McgregorV16,KapralovKS14}. In this
paper, we make progress on this question from both upper bound and lower bound ends.

\subsection{Our Results} \label{sec:current}

\paragraph{Upper bounds.} We prove that computing an $\alpha$-approximate estimate of matching size is strictly easier than finding an $\alpha$-approximate
matching. Formally, 

\begin{theorem}\label{thm:upper-intro} 
  There exist single-pass streaming algorithms that for any $2 \leq \alpha \leq \sqrt{n}$, w.h.p.\footnote{We use
    w.p. and w.h.p. to abbreviate with probability and with high probability, respectively.}, output an $\alpha$-approximation of the
  maximum matching size in dynamic streams using $\Ot(n^2/\alpha^4)$ and in insertion-only streams using $\Ot(n/\alpha^2)$ space,
  respectively.
\end{theorem}

The algorithms in Theorem~\ref{thm:upper-intro} are the first algorithms that outperform (by a factor of $\alpha$), respectively, the
\emph{optimal} $\Ot({n^2 / \alpha^3})$-space algorithm in dynamic streams, and the \emph{optimal} $\Ot(n/\alpha)$-space algorithm in
insertion-only streams for finding an $\alpha$-approximate matching. This provides the first non-trivial separation between approximate
estimation and approximate computation of matchings in both dynamic and insertion-only streams. 

\paragraph{Lower bounds.} 
Our first lower bound result concerns computing an $\alpha$-approximation of the maximum matching size in dynamic streams for any
$\alpha \geq 1$, \emph{not necessarily a constant}. 

\begin{theorem}\label{thm:lower-alpha-intro}
  Any (randomized) single-pass streaming algorithm that computes an $\alpha$-approximation of maximum matching size with a constant
  probability in dynamic streams requires $\Omega(\sqrt{n}/\alpha^{2.5})$ bits of space. This bound holds even if the input graph is both \emph{sparse} and has
  \emph{arboricity}\footnote{A graph $G$ has arboricity $\nu$ if the set of edges in $G$ can be partitioned into at most $\nu$ forests.}
  $O(\alpha)$.  Moreover, if the input graph is allowed to be \emph{dense}, then $\Omega(n/\alpha^2)$ bits of space is necessary.
\end{theorem}

The lower bounds in Theorem~\ref{thm:lower-alpha-intro} are the first non-trivial space lower bounds for
\emph{super-constant} approximation algorithms for matching size estimation. Obtaining space lower bounds for $\polylog{(n)}$-approximation
of matching size has been posed as an open problem by Kapralov~\etal\cite{KapralovKS14}, who also mentioned that ``existing techniques do not seem to lend easily to answer this question and it will be very useful (quite possibly for other related problems) to develop tools needed to make progress on this front''. Our results in
Theorem~\ref{thm:lower-alpha-intro} make progress on this question in dynamic streams.

An interesting aspect of our lower bound in Theorem~\ref{thm:lower-alpha-intro} is that it holds even for bounded arboricity graphs. There
is an active line of research on estimating matching size of bounded arboricity graphs in graph
streams~\cite{ChitnisCEHMMV16,BuryS15,EsfandiariHLMO15,McgregorV16}, initiated by Esfandiari~\etal~\cite{EsfandiariHLMO15}. The state-of-the-art is an $O(1)$-approximation in $\Ot(n^{4/5})$ space for dynamic streams in
bounded-arboricity graphs~\cite{ChitnisCEHMMV16,BuryS15,McgregorV16}.

Our second lower bound result concerns computing a $(1+\eps)$-approximation of the maximum matching size in both insertion-only streams and
in dynamic streams.  In the following, let $\RS{n}$ denote the maximum number of edge-disjoint \emph{induced matchings} of size $\Theta(n)$
in any $n$-vertex graph (see Section~\ref{sec:rs-graphs}).

\begin{theorem}\label{thm:lower-eps-intro}
  Any (randomized) single-pass streaming algorithm that with a constant probability outputs a $(1+\eps)$-approximation of the
  maximum matching size in insertion-only streams requires $\RS{n} \cdot n^{1-O(\eps)} $ space. The lower bound improves
  to $n^{2-O(\eps)}$ for dynamic streams. 
\end{theorem}

Since $\RS{n}$ is known to be at least $n^{\Omega{(1/\log\log n})}$~\cite{FischerLNRRS02}, Theorem~\ref{thm:lower-eps-intro} immediately
implies that no $\Ot(n \cdot \poly(1/\eps))$-space algorithm can output a $(1+\eps)$-approximation of matching size in
insertion-only streams. Interestingly, it is known that by allowing \emph{multiple passes} over the stream, a $(1+\eps)$-approximate matching (as opposed to only its size) can be found in 
$\Ot(n \cdot \poly(1/\eps))$ space, even in dynamic streams and even for the weighted version of the problem~\cite{AhnG15,AG13} (see also~\cite{M14}).

Our lower bounds in Theorem~\ref{thm:lower-eps-intro} are the first \emph{super linear} (in $n$) space lower bounds for estimating matching
size in graph streams.  An interesting implication of these lower bounds is that while the problem of matching size estimation is provably
easier than the problem of finding an approximate matching (by Theorem~\ref{thm:upper-intro}), the space complexity of the two problems
starts to converge together as the accuracy desired in the computation approaches near-optimality.

\paragraph{Schatten $p$-norms.} The \emph{Schatten $p$-norm} of a matrix $A$ is defined as the $\ell_p$-norm of the vector of the singular
values of $A$ (see~\cite{LiW16} for more detail); in particular, the case of $p=0$ corresponds to the \emph{rank} of the matrix
$A$. Schatten norms and rank computation have been previously studied in the streaming and sketching
models~\cite{ClarksonW09,BuryS15,LiW16,LiNW14a,LiW16a,LiSWW14}.  It is shown that exact computation of matrix rank in data streams requires
$\Omega(n^2)$ space~\cite{ClarksonW09,LiSWW14} (even allowing multiple passes), and $(1+\eps)$-approximation requires $n^{1-O(\eps)}$
space~\cite{BuryS15}; the latter result was recently extended to all Schatten $p$-norms for \emph{odd} values of $p$~\cite{LiW16}.

It is well-known that computing the maximum matching size 
is equivalent to computing the rank of the Tutte matrix~\cite{Tutte47,Lovasz09}. Consequently, all our lower bounds stated for matching size estimation also hold for matrix rank computation.   
This in particular implies an $\Omega(\sqrt{n})$ space lower bound for \emph{any constant} approximation of rank in \emph{sparse} matrices and a near-optimal $n^{2-O(\eps)}$ space lower bound 
for $(1+\eps)$-approximation in \emph{dense} matrices, answering an open question of Li and Woodruff~\cite{LiW16}.

\section{Preliminaries}\label{sec:prelim}

\paragraph{Notation.} For any graph $G$, $\opt(G)$ denotes the maximum matching \emph{size} in $G$. 
We use bold face letters to represent random variables. For any random variable $\bX$, $\supp{\bX}$ denotes its support set. We define
$\card{\bX}: = \log{\card{\supp{\bX}}}$.  For any $k$-dimensional tuple $X = (X_1,\ldots,X_k)$ and any $i \in [k]$, we define
$X^{<i}:=(X_1,\ldots,X_{i-1})$, and $X^{-i}:=(X_1,\ldots,X_{i-1},X_{i+1},\ldots,X_{k})$.


\paragraph{Total Variation Distance.} For any two distributions $\mu$ and $\nu$ with the same support $\Omega$ where $\card{\Omega}$ is
finite, the \emph{total variation distance} between $\mu$ and $\nu$, denoted by $\tvd{\mu}{\nu}$, is given by
$\max_{\Omega' \subseteq \Omega} \paren{\mu(\Omega') - \nu(\Omega')} = \frac{1}{2} \sum_{x \in \Omega} \card{\mu(x) - \nu(x)}$.  We use the
following well-known fact in our proofs.

\begin{fact}\label{fact:tvd}
  Suppose we want to distinguish between two probability distributions $\mu$ and $\nu$ given one sample from one of the two; then the
  best probability of success is $\frac{1}{2} + \frac{\tvd{\mu}{\nu}}{2}$.
\end{fact}


\subsection{\rs graphs}\label{sec:rs-graphs}
For any graph $G$, a matching $M$ of $G$ is an \emph{induced matching} iff for any
two vertices $u$ and $v$ that are matched in $M$, if $u$ and $v$ are not matched to each other, then there is no edge between $u$ and $v$ in $G$.

\begin{definition}[\rs graph]\label{def:rs-graph}
  A graph $G$ is an $(r,t)$-\emph{\rs graph} (or $(r,t)$-RS graph for short), iff the set of edges
  in $G$ consists of $t$ pairwise disjoint induced matchings $M_1,\ldots,M_t$, each of size $r$.
\end{definition}

RS graphs, first introduced by Ruzsa and Szemer\'{e}di~\cite{RuszaS78}, have been extensively studied as they arise naturally in property
testing, PCP constructions, additive combinatorics, etc. (see,
e.g.,~\cite{TaoV06,HastadW03,FischerLNRRS02,BirkLM93,AlonMS12,GoelKK12,Alon02,AlonS06,FoxHS15}).  These graphs are of interest typically
when $r$ and $t$ are large relative to the number of vertices in the graph.

One particularly interesting range of the parameters is when $r = \Theta(n)$~\cite{FischerLNRRS02,Fox11,FoxHS15}, i.e., when the induced
matchings are of linear size. We use the notation $\RS{n}$ to denote the \emph{largest} possible value for the parameter $t$ such that an
$(r,t)$-RS graph on $n$ vertices with $r = \Theta(n)$ exists. It is a major open problem to determine the asymptotic of
\RS{n}~\cite{Fox11,Gowers01,FoxHS15}, but currently there is a huge gap between existing upper and lower bounds for $\RS{n}$.  In
particular, it is known that for any constant $c < 1/4$, a $(c\cdot n,t)$-RS graph with $t = n^{\Omega(1/\log\log{n})}$
exists~\cite{FischerLNRRS02} (see also~\cite{GoelKK12}). However, the best known upper bound only shows that for $(c \cdot n,t)$-RS graphs, where $c$ is any constant
less than $1/4$, $t$ is upper bounded by ${n \over \log^{(x)}{n}}$, with $x = O(\log{1 \over c})$, ($\log^{(x)}(n)$ denotes the $x$-fold
iterative logarithm of $n$)~\cite{Fox11}. Slightly better upper bounds are known for large values of $c$; in particular, it is shown in~\cite{FoxHS15}
that for $1/5 < c < 1/4$, $t = O(n/\log{n})$. We refer the interested reader to~\cite{FoxHS15,AlonMS12} for more on the history of \rs
graphs and to~\cite{AlonMS12,GoelKK12} for their application to different areas of computer science, including proving lower bounds for
streaming algorithms.

Obtaining $(r,t)$-RS graphs for $r=\Theta(n)$ and $t = n^{\eps}$ (for some constant $\eps > 0$) seems to be out of the scope of the
state-of-the-art techniques; however, Alon~\etal~\cite{AlonMS12} provide a surprising construction of (very) dense RS graphs when we allow
$r$ to be just \emph{slightly sublinear}: there are $(r,t)$-RS graphs on $n$ vertices with parameters $r = n^{1-o(1)}$ and
$r \cdot t = {{n \choose 2}} - o(n^2)$~\cite{AlonMS12}.  While our lower bound for insertion-only streams requires the use of $(r,t)$-RS
graphs with $r = \Theta(n)$ (hence naturally leads to a dependence on $\RS{n}$), for our lower bound for dynamic streams it suffices to work
with RS graphs with $r = n^{1-o(1)}$ and hence we can directly use the construction of~\cite{AlonMS12} (hence avoiding dependency on
$\RS{n}$).

\subsection{Tools from Information Theory}\label{sec:info}
We briefly review some basic concepts from information theory needed for establishing our lower bounds. For a
broader introduction to the field, we refer the reader to the excellent text by Cover and
Thomas~\cite{ITbook}.

In the following, we denote the \emph{Shannon Entropy} of a random variable $\bA$ by
$H(\bA)$ and the \emph{mutual information} of two random variables $\bA$ and $\bB$ by
$I(\bA;\bB) = H(\bA) - H(\bA \mid \bB) = H(\bB) - H(\bB \mid \bA)$. If the distribution
$\dist$ of the random variables is not clear from the context, we use $H_\dist(\bA)$
(resp. $I_{\dist}(\bA;\bB)$). We use $H_2$ to denote the binary entropy function where for any real number $0
< \delta < 1$, $H_2(\delta) = \delta\log{\frac{1}{\delta}} + (1-\delta)\log{\frac{1}{1-\delta}}$. 
We know that $0 \leq H(\bA) \leq \card{\bA}$ and equality holds iff $\bA$ is uniform on its support. Similarly, $I(\bA;\bB) \geq 0$ and equality holds
iff $\bA$ and $\bB$ are independent of each other. 

We use the following basic properties of entropy and mutual information (proofs can be
found in~\cite{ITbook}, Chapter~2).
\begin{fact}\label{fact:it-facts}
  Let $\bA$, $\bB$, $\bC$ be random variables. 
  \begin{enumerate}
  \item \label{part:cond-reduce} \emph{Conditioning reduces the entropy}:
    $H(\bA \mid \bB,\bC) \leq H(\bA \mid \bB)$; equality holds iff $\bA$ and $\bC$
    are independent conditioned on $\bB$.
  \item \label{part:ent-chain-rule} \emph{Chain rule for entropy}: $H(\bA,\bB) =H(\bA ) + H(\bB \mid \bA)$.
  \item \label{part:info-chain-rule} \emph{Chain rule for mutual information}: $I(\bA,\bB ;
    \bC) = I(\bA ; \bC) + I( \bB;\bC \mid \bA)$.
  \item \label{part:info-sub-additivity} \emph{Conditional sub-additivity of mutual information}: if $\bA_1,\bA_2, \ldots,\bA_t$ are
    conditionally independent given $\bB$, then $I(\bA_1,\bA_2, \ldots,\bA_t ; \bB) \leq \sum_{i=1}^t I(\bA_i ; \bB)$.
  \item \label{part:info-super-additivity} \emph{Conditional super-additivity of mutual information}: if $\bA_1,\bA_2, \ldots,\bA_t$ are
    conditionally independent given $\bC$, then $I(\bA_1,\bA_2, \ldots,\bA_t ; \bB \mid \bC) \geq \sum_{i=1}^t I(\bA_i ; \bB \mid \bC)$.
  \end{enumerate}
\end{fact}

The following claim (Fano's inequality) states that if a random variable $\bA$ can be used to estimate the value
of another random variable $\bB$, then $\bA$ should ``consume'' most of the entropy of
$\bB$.

\begin{claim}[Fano's inequality]\label{clm:fano}
For any binary random variable $\bB$ and any (possibly randomized) function f that
predicts $\bB$ based on $\bA$, if $\Pr(f(\bA) \neq \bB) = \delta$, then $H(\bB \mid \bA)
\leq H_2(\delta)$.
\end{claim}


Finally, we prove the following auxiliary lemma that allows us to decompose any random variable with high entropy to a convex combination of relatively small number of near uniform distributions plus a low probability ``noise term''. 

\begin{lemma}\label{lem:aux-high-ent}
  Let $\bX \sim \dist$ be a random variable on $\set{0,1}^{n}$ such that $H(\bX) \geq n-\Delta$. For any $\eps > 0$, 
  there exist $k+1$ distributions ${\mu_0,\mu_1,\ldots,\mu_k}$ on $\set{0,1}^{n}$,
  along with $k+1$ probabilities $p_0, p_1, \ldots, p_k$ ($\sum_i p_i = 1$) for some $k=O(n/\eps)$, such that
  $\dist = \sum_{i} p_i \cdot \mu_i$, $p_0 = O(\eps)$, and for any $i \geq 1$,
	\begin{enumerate}
		\item $\log{\card{\supp{\mu_i}}} \geq {n-\frac{\Delta}{\eps}-\log{\Theta(\frac{n}{\eps})}}$. 
		\item $\tvd{\mu_i}{U_i} = O(\eps)$, where $U_i$ denotes the \emph{uniform distribution} on $\supp{\mu_i}$.
	\end{enumerate} 
\end{lemma}
\begin{proof}
  Partition the support of $\dist$ into $k'$ sets $S_0,S_1,\ldots,S_{k'}$ for $k' = \Theta(n/\eps)$ such that $S_0$ contains every element
  $a \in \set{0,1}^n$ where $\Pr(\bX=a) < 2^{-2n}$, and for each $i \geq 1$, $S_i$ contains every element $a$ where
  $(1+\eps)^{-(i+1)} \leq \Pr(\bX=a) < (1+\eps)^{-i}$.  We say that a set $S_i$ is \emph{large} if
  $\card{S_i} \geq 2^{(n-\frac{\Delta}{\eps} - \log{\Theta(\frac{n}{\eps})})}$ and is otherwise \emph{small}. Let $\FL$ (resp. $\FS$) denote the set of all elements that
  belong to a large set (resp. a small set). Moreover let $k$ be the number of large sets, and, without loss of generality, assume $S_1,\ldots,S_k$ are 
  these large sets. 

  We define the $k+1$ distributions in the lemma statement as follows. Let $\mu_0$ be the distribution $\dist$ conditioned on $\bX$ being in $S_0 \cup \FS$
  (i.e., $S_0$ and elements from small sets), and let $p_0 = \Pr_{\dist}(\bX \in \FS \cup S_0)$; for each $i \ge 1$, let
  $\mu_i$ be the distribution $\dist$ conditioned on $\bX$ being in $S_i$ (i.e., the $i$-th large set) and let $p_i = \Pr_{\dist}(\bX \in S_i)$.

  By construction, the described distributions satisfy  $\dist = \sum_{i} p_i \cdot \mu_i$. Moreover, for each $i \ge 1$, since the
  support $S_i$ of $\mu_i$ is a large set, we have $\log \card{\supp{\mu_i}} \geq {n-\frac{\Delta}{\eps} - \log{\Theta(\frac{n}{\eps})}}$; since each element $a$ in $\supp{\mu_i}$ has
  $\Pr_{\dist}(\bX=a) \in [(1+\eps)^{-(i+1)}, (1+\eps)^{-i})$, it is straightforward to verify that $\tvd{\mu_i}{U_i} = O(\eps)$. Hence it only
  remains to argue that $p_0 = O(\eps)$.
	 
	 It is easy to see that $\Pr_{\dist}(\bX \in S_0) = o(1)$ and therefore in the following we prove that $\Pr_{\dist}(\bX \in \FS) = O(\eps)$. 
	Let $\bZ \in \set{0,1}$ be a random variable that denotes whether $\bX$ chosen from $\dist$ belongs to $\FL$ or $\FS$. 
	We have, 
	\begin{align}
		H(\bX \mid \bZ) \geq H(\bX) - H(\bZ) \geq H(\bX) - 1 \geq n-\Delta-1 \label{eq:H-dist}
	\end{align}
	where the first inequality is by chain rule of entropy (\itfacts{ent-chain-rule}). Moreover, since the total number of elements belonging to small sets is at most 
	${\Theta(n/\eps) \cdot 2^{\paren{n-\frac{\Delta}{\eps} - \log{\Theta(\frac{n}{\eps})}}}} = 2^{n-\frac{\Delta}{\eps}}$, 
	\begin{align*}
          H(\bX \mid \bZ) &= \Pr(\bZ = 0) \cdot H(\bX \mid \bZ = 0) + \paren{1-\Pr(\bZ = 0)} \cdot H(\bX \mid \bZ = 1) \\
                          &\leq \Pr(\bZ = 0) \cdot \log{\paren{2^{n-\frac{\Delta}{\eps}}}} + \paren{1-\Pr(\bZ=0)} \cdot \log\paren{2^{n}}  \tag{since $H(\bA) \leq \card{\bA}$ for any random variable $\bA$} \\
                          &= n - \Pr(\bZ=0) \cdot \paren{\frac{\Delta}{\eps}} 
	\end{align*}
	and consequently, if $\Pr(\bZ = 0) > 2\eps$, then $H(\bX \mid \bZ) < n - 2\Delta < n-\Delta-1$, a contradiction to
        Eq~(\ref{eq:H-dist}). This finalizes the proof that $p_0 = O(\eps)$.
\end{proof}


\subsection{Communication Complexity and Information Complexity}\label{sec:cc-ic}
Communication complexity and information complexity play an important role in our lower bound proofs. 
We now provide necessary definitions for completeness.

\paragraph{Communication complexity.} Our lowers bounds for streaming algorithms
are established through communication complexity lower bounds.  Here, we briefly
provide some relevant background; for a more detailed treatment of
communication complexity, we refer the reader to the excellent text by Kushilevitz and
Nisan~\cite{KN97}.

We focus on two models of communication, namely, the \emph{two-player one-way communication} model, 
and the \emph{multi-party number-in-hand simultaneous message passing model} (SMP).  

\paragraph{One-way Communication Model.} Let $P$ be a relation with
domain $\mathcal{X} \times \mathcal{Y} \times \mathcal{Z}$.  Alice receives an input $X
\in \mathcal{X}$ and Bob receives $Y \in \mathcal{Y}$, where $(X,Y)$ are chosen from a
joint distribution $\dist$ over $\mathcal{X} \times \mathcal{Y}$.  In addition to private randomness, the players also have an
access to a shared public tape of random bits $R$. Alice sends a single message $M(X,R)$
and Bob needs to output an answer $Z := Z(M(X,R),Y,R)$ such that $(X,Y,Z) \in P$.

We use $\Prot$ to denote a protocol used by the players. Unless specified otherwise, we
always assume that the protocol $\Prot$ can be randomized (using both public and
private randomness), \emph{even against a prior distribution $\dist$ of inputs}. For any
$0 < \delta < 1$, we say $\Prot$ is a $\delta$-error protocol for $P$, if the probability that
 for \emph{any} input $(X,Y)$, Bob outputs some $Z$ where $(X,Y, Z) \in P$ is at least
$1-\delta$ (over the randomness of the protocol $\Prot$).  

The \emph{communication cost} of a one-way protocol $\Prot$ for a problem $P$ on an input
  distribution $\dist$, denoted by $\norm{\Prot}$, is the worst-case size of the message
  sent from Alice to Bob in the protocol $\Prot$, when the inputs are chosen from the distribution
  $\dist$. \emph{Communication complexity} $\CCO{P}{\dist}{\delta}$ of a
  problem $P$ with respect to a distribution $\dist$ is the minimum communication cost of any one-way protocol
  $\Prot$ that is required to solve $P$ on \emph{every instance} w.p. at least $1-\delta$.  

\paragraph{SMP Communication Model.} Let $P$ be a $(k+1)$-ary relation with domain $\mathcal{X}_1\times \ldots \times \mathcal{X}_k \times \mathcal{Z}$. 
In the SMP communication model, $k$ players $\PS{1},\ldots,\PS{k}$ recieve inputs $X_1,\ldots,X_k$, jointly distributed 
according to a prior distribution $\dist$ over $\mathcal{X}_1 \times \ldots \times \mathcal{X}_k$. In addition to private randomness, the players also have an
access to a shared public tape of random bits $R$. Each of the players simultaneously sends a single message $M_j(X_j,R)$ to an external party called
the \emph{referee} and referee needs to output an answer $Z:= Z(M_1(X_1,R),\ldots,M_k(X_k,R),R)$ such that $(X_1,\ldots,X_k,Z) \in P$.  

Similar to the one-way communication model, we let $\Prot$  denote the protocol used by the players and define $\delta$-error protocols for $P$ over
a distribution $\dist$ analogously. The \emph{communication cost} of a SMP protocol $\Prot$ for a problem $P$ on an input
  distribution $\dist$, denoted by $\norm{\Prot}$, is the sum of the worst-case size of the messages sent by players to the 
  referee, i.e., $\norm{\Prot}:= \sum_{i \in [k]} \card{M_i}$, when the inputs are chosen from the distribution
  $\dist$.  \emph{Communication complexity} $\CCS{P}{\dist}{\delta}$ in SMP model is defined the same as in the one-way communication model. 
  
\begin{remark}\label{rem:referee}
 To facilitate our proofs, we sometimes need to give the referee an auxiliary input as well, which is jointly distributed with the 
 input of the $k$ players. The referee's answer then would be a function of the $k$ messages he receives as well as his input. As a convention, we typically ignore
 this artificial feature of the model and only include it implicitly. 
\end{remark}

\paragraph{Information Complexity.} 
There are several possible definitions of information complexity of a communication problem that have been considered depending on
the application (see, e.g.,~\cite{Bar-YossefJKS02,BarakBCR10,ChakrabartiSWY01,Bar-YossefJKS02-S}). In this paper, we use the notion of (external) \emph{information cost} of a protocol. 
Roughly speaking, information cost of a one-way or SMP protocol is the average amount of information one can learn about the input of the players that are sending the messages by observing the transcript of 
the protocol. 

More formally, the \emph{information cost} of a one-way protocol $\Prot$ with respect to a distribution $\dist$ is $\ICost{\Prot}{\dist} =  I_\dist(\bProt;\bX)$, 
where $\bX \sim \dist$ is the random variable for the input to Alice, $\bProt := \bProt(\bX)$ is the random variable denoting the message sent
from Alice to Bob in the protocol $\Prot$, \emph{concatenated} with the \emph{public} randomness $\bR$ used by $\Prot$. 
The \emph{information complexity} $\ICO{P}{\dist}{\delta}$ of $P$ with respect to a distribution $\dist$ is  the minimum $\ICost{\Prot}{\dist}$ taken over all one-way protocols
$\Prot$ that are required to solve $P$ on \emph{every instance} w.p. at least $1-\delta$. 

Similarly, the information cost of a SMP protocol is defined as $\sum_{j=1}^{k} I_\dist(\bProt_j;\bX_1,\ldots,\bX_k)$, where 
$\bX_i$ denotes the input of the player $\PS{i}$ and $\bProt_i := \bProt_i(X_i)$ denotes the message sent from the player $\PS{i}$ to
the referee \emph{concatenated} with the public randomness $\bR$ used by $\Prot$. \emph{Information complexity} $\ICS{P}{\dist}{\delta}$ of a problem $P$ in
the SMP communication model can also be defined analogous to the one-way communication model. 

\begin{remark}\label{rem:trivial-dist}
	The requirement in the above definitions that $\Prot$ is correct \emph{everywhere}, even outside the support of the distribution $\dist$ is \emph{crucial}: we analyze our lower bounds on 
	distributions that are ``trivial'' and the only reason that these lower bounds are meaningful (i.e., are non-zero) is that these protocols are required to succeed \emph{uniformly}. 
\end{remark}


The following well-known proposition (see, e.g.,~\cite{ChakrabartiSWY01}) relates communication complexity and information complexity. 
\begin{proposition}\label{prop:cc-ic}
  For every $0 < \delta < 1$ and every distribution $\dist$: 
  \begin{align*}
  (i)~\CCO{P}{\dist}{\delta} &\geq \ICO{P}{\dist}{\delta} & (ii)~\CCS{P}{\dist}{\delta} \geq \ICS{P}{\dist}{\delta}
  \end{align*}
\end{proposition}

\begin{proof}
 We only prove this for the SMP model; the result for the one-way model can be proven similarly. 
 Let $\Prot$ be a SMP protocol with the minimum communication complexity for $P$ on $\dist$ and $\bR$ denote the public randomness of $\Prot$;
 we have, 
	\begin{align*}
		\ICS{P}{\dist}{\delta} &\leq \ICost{\Prot}{\dist} = \sum_{i=1}^{k} I_{\dist}(\bProt_i; \bX) = \sum_{i=1}^{k} I_{\dist}(\bProt_i,\bR; \bX) \tag{$\bProt_i$ contains both the message of
		 player $\PS{i}$ and the public randomness $\bR$} \\
		&= \sum_{i=1}^{k} I_{\dist}(\bR;\bX) + I_{\dist}(\bProt_i;\bX,\bR) \tag{chain rule of mutual information (\itfacts{info-chain-rule})} \\
		&=  \sum_{i=1}^{k} I_{\dist}(\bProt_i;\bX,\bR) = \Ex_{R \sim \bR}\Bracket{\sum_{i=1}^{k} I_{\dist}(\bProt_i ; \bX \mid \bR = R)} 
		\tag{$\bR \perp \bX$ and hence $I(\bR;\bX) = 0$} \\
		&\leq  \Ex_{R \sim \bR}\Bracket{\sum_{i=1}^{k} H_{\dist}(\bProt_i ; \bX \mid \bR = R)} \leq \Ex_{R \sim \bR}\Bracket{\sum_{i=1}^{k} \card{\bProt^{R}_i}}
		\tag{$\bProt^{R}_i$ is the message sent from $\PS{i}$ and is equal to $\bProt_i$ conditioned on $\bR = R$} \\
		&\leq \norm{\Prot} = \CCS{P}{\dist}{\delta}
 	\end{align*}
\end{proof}

\emph{Connection to Streaming:} We conclude this section by pointing out the connection between the communication models defined in this section and the streaming setting. It is a standard fact that 
any streaming algorithm directly works as a one-way communication protocol and hence lower bounds in the one-way communication model also imply the same bounds on the space complexity of 
streaming algorithms in \emph{insertion-only} streams. Recent results of~\cite{AiHLW16,LiNW14} prove a similar situation for the SMP model and \emph{dynamic} streams: communication complexity lower bounds
in SMP model imply space lower bounds for dynamic streams. In particular, communication complexity of a \emph{$k$-player} problem in the SMP model is at most $k$ \emph{times} the space complexity
of the same problem in dynamic streams. 

\subsection{The Boolean Hidden Hypermatching Problem}\label{sec:bhh-prelim}
 We shall use the following communication problem first studied by~\cite{VerbinY11} in proving our lower bounds.

\begin{definition}[\textbf{Boolean Hidden Hypermatching}, $\BHH{n}{t}$]\label{def:bhh} The \emph{Boolean Hidden Hypermatching problem} is a 
one-way communication problem in which Alice is given a boolean vector $x \in \set{0,1}^{n}$ where $n = 2kt$ (for some integer $k \geq 1$)
and Bob gets a \emph{perfect $t$-hypermatching} $M$ on $n$ vertices, and a boolean vector $w \in \set{0,1}^{n/t}$. Let $Mx$ denote the length $n/t$ 
boolean vector $(\bigoplus_{1\leq i\leq t} x_{M_{1, i}}, \ldots, \bigoplus_{1\leq i\leq t} x_{M_{n/t, i}})$ where $\set{M_{1, 1,}, \ldots, M_{1, t}}, \ldots, \set{M_{n/t, 1}, \ldots, M_{n/t, t}}$ are the edges of $M$. 
It is promised that either $Mx = w$ or $Mx = \negative{w}$. The goal of the problem is for Bob to output
\Yes when $Mx = w$ and \No when $Mx = \negative{w}$ ( $\oplus$ stands for addition modulo $2$).
\end{definition}

The special case of this problem where $t=2$ is referred to as the \emph{Boolean Hidden Matching} problem, \BHM{n}, and was originally
introduced by Gavinsky~\etal~\cite{GKKRW07} who established an $\Omega(\sqrt{n})$ lower bound on its
one-way communication complexity. This lower bound was extended to $\Omega(n^{1-1/t})$ for the more general \BHH{n}{t} problem by
Verbin and Yu~\cite{VerbinY11} (see Section~\ref{sec:bhh} for more details). We further extend this 
result and establish a matching lower bound on the \emph{information complexity} of $\BHH{n}{t}$ (see Theorem~\ref{thm:bhh-ic}).

For our purpose, it is more convenient to work with a special case of the $\BHH{n}{t}$ problem, namely $\BHHZ{n}{t}$ where the vector $w = 0^{n/t}$ and hence the goal of 
Bob is simply to decide whether $Mx =0^{n/t}$ (\Yes case) or $Mx =1^{n/t}$ (\No case). We define $\BHMZ{n} := \BHHZ{n}{2}$ (similar to $\BHM{n}$).
It is known that (see, e.g.~\cite{VerbinY11,BuryS15,LiW16}) any instance of the original $\BHH{n}{t}$ problem 
can be reduced to an instance of $\BHHZ{2n}{t}$ \emph{deterministically} and with \emph{no communication} between the players. 
This allows for extending the communication and information complexity lower bounds of $\BHH{n}{t}$ to $\BHHZ{2n}{t}$ problem (see Corollary~\ref{cor:bhh}). 


\paragraph{$\BHHZ{n}{t}$ and Matching Size Estimation.} 
  The $\BHHZ{n}{t}$ problem has been used previously in~\cite{EsfandiariHLMO15,BuryS15} to prove lower bounds for estimating matching size in data streams. 
We now briefly describe this connection. 

 The following reduction was first proposed by~\cite{BuryS15}. Given an instance $(x,\HM)$\footnote{In order to distinguish between matchings and hypermatchings, when not clear from the context, we use
 $\HM$ instead of $M$ to denote a hypermatching.} of $\BHHZ{n}{t}$, we create a graph $G(V \cup W,E)$ with $\card{V} = \card{W} = n$ as follows: 
 \begin{itemize}
 	\item For any $x_i = 1$, Alice adds an edge between $v_i$ and $w_i$ to $E$. 
 	\item For any hyperedge $\bm{e}$ in the $t$-hypermatching $\HM$, Bob adds to $E$  a clique between the vertices $w_i$ where $i$ is incident on $\bm{e}$. 
 \end{itemize}
 
 The following claim, proven originally by~\cite{BuryS15}, establishes the correctness of this reduction. 
 For the sake of completeness, we provide a simple proof this claim here. 
 
 \begin{claim}[\!\!\cite{BuryS15}]\label{clm:bhh-matching}
 	Suppose $G(V \cup W,E)$ is the graph obtained from an instance $(x,\HM)$ of $\BHHZ{n}{t}$ (for an \emph{even integer} $t$) with the property that $\norm{x}_0 = n/2$; 
 	\begin{itemize}
 		\item if $Mx = 0^{n/t}$ (i.e., \Yes case), then $\mu(G) = \frac{3n}{4}$. 
 		\item if $Mx = 1^{n/t}$ (i.e., \No case), then $\mu(G) = \frac{3n}{4} - \frac{n}{2t}$. 
 	\end{itemize} 
 \end{claim}
 \begin{proof}
   Denote by $\Mstar$ a maximum matching in $G$.  Since the vertices in $V$ all have degree one, without loss of generality, we can assume
   all edges in $V \times W$ belong to $\Mstar$, and we only need to consider the maximum matching size between the remaining vertices.
   Since the remaining vertices in $V$ all have degree $0$, we only need to consider the remaining vertices in $W$ (and $n/2$ vertices in
   $W$ remains since $\norm{x}_0 = {n \over 2}$). 

   In the \Yes case, for each hyperedge $\bm{e}$, the clique created by $\bm{e}$ has $t$ vertices, and even number of these vertices will be
   matched by edges in $V \times W$. Since $t$ is even, \emph{even} number of the vertices of the clique remain.  Since there is still a
   clique between these remaining vertices, there is a matching that matches all of them. Therefore, the total matching size is
   ${n \over 2} + \half \cdot {n \over 2} = {3n \over 4}$.

   In the \No case, for each hyeredge $\bm{e}$, the clique created by $\bm{e}$ has \emph{odd} number of vertices remained. Therefore, for
   every hyperedge, one vertex will be left unmatched. Since there are ${n \over t}$ hyperedges, ${n \over t}$ of the remaining vertices
   will be left unmatched, hence the total matching size is
   ${n \over 2} + \half \paren{{n \over 2} - {n \over t}} = {3n \over 4} - {n \over 2t}$.
 \end{proof}

\section{Technical Overview}\label{sec:tech}

\subsection{Lower Bounds} Our lower bounds are obtained by establishing communication complexity lower bounds in the one-way model 
(for insertion-only streams) and in the SMP model (for dynamic streams). 
We prove our lower bounds for sparse graphs (first part of Theorem~\ref{thm:lower-alpha-intro}) and dense 
graphs (Theorem~\ref{thm:lower-eps-intro} and second part of Theorem~\ref{thm:lower-alpha-intro}) using conceptually different techniques; we elaborate below on each case separately. 

\paragraph{Sparse graphs.} We prove this lower bound by analyzing the following $k$-player problem, referred to as the \emph{sparse matching size estimation} (\sms) problem, in the SMP 
 model: each player $\PS{i}$ (for $i \in [k]$) is given a matching $M_i \subseteq E$ in a \emph{sparse} graph $G(V_S \cup V_P,E)$ with $\card{V_P} = \Theta(k) \cdot \card{V_S}$;
 think of vertices in $V_S$ as \emph{shared} vertices that appear in the input of every player and 
vertices in $V_P$ as \emph{private} vertices that appear in the input of only a single player (the partition $V_S$ and $V_P$ is \emph{not} known to the players). 
In the \Yes case, the end-points of every edge are either both shared or both private such that $\opt(G) = \Theta(V_P)$, and in the \No case, 
every edge has one shared end-point and one private end-point, hence $\opt(G) = \Theta(V_S)$. 



We can interpret the setup in the \sms problem as follows. For any player $\PS{i}$ with the matching $M_i$, define a binary vector $x_{i}$ over the set $V(M_i)$ of vertices incident on $M_i$: 
for any $v \in V(M_i)$, $x_i(v) = 1$ if the vertex $v$ is a shared vertex and  $x_i(v) = 0$ otherwise. The vector $x_i$ for player \PS{i} can be identified uniquely given the set of vertices in $V(M_j)$ of any other
player $j \neq i$.  
Now, in the \Yes case (resp. the \No case) of the \sms problem, for any matching $M_i$ and any two vertices $u,v \in V(M_i)$, $x_i(u) \oplus x_i(v) = 0$ (resp. $x_i(u) \oplus x_i(v) = 1$). 
One may notice that this setup is quite similar to the \bhmz problem in the one-way model described in Section~\ref{sec:prelim}. 
Indeed, we ultimately prove a lower bound on 
the simultaneous communication complexity of our \sms problem using the $\Omega(\sqrt{n})$ lower bound of \bhmz problem~\cite{GKKRW07}. However, there is
an inherent difficulty in performing such a reduction that we elaborate on next. Addressing this challenge results in a rather non-standard and protocol-specific reduction of a 
simultaneous multi-player problem to a two-player one-way problem, which is one of our central technical contributions. 

A standard technique in proving communication lower bounds for multi-player problems is \emph{symmetrization}~\cite{PhillipsVZ12}; here, one reduces a $2$-player problem to a $k$-player problem by 
letting Bob play the role of one of the $k$ players and Alice play the role of the remaining $(k-1)$ players. This technique is used (both explicitly and implicitly) in many 
known lower bounds for \emph{finding} approximate matchings in different multi-player communication models~\cite{Kapralov13,Konrad15,AssadiKLY16,HuangRVZ15}. 
The success of this technique in these cases can be mostly attributed to the fact that in finding an approximate matching, every player is responsible for
communicating the set of edges in \emph{his input} that belongs to a maximum matching in the final graph; in other words, the message communicated by a player is typically 
\emph{not} helping in finding the edges of another player. 

In contrast, in matching \emph{size} estimation, the players only need to (together) convey a \emph{signal} about whether their common input is a \Yes instance or a \No instance. In particular, a small number of players already have enough information to distinguish between the large and small matching size cases;  
for example, in the \sms problem, any two players together have sufficient information 
to solve the problem completely. Indeed, the two players $\PS{i}$ and $\PS{j}$ can identify the set of shared vertices (and hence the vector $x_i$) 
and then simply check the {parity} of one arbitrary edge in $M_i$ using $x_i$, to distinguish between the two cases. This implies that
no matter how we split the role of the $k$ players between Alice and Bob, Alice already gains enough information from the distribution 
to solve the underlying \bhmz instance. 

To circumvent this issue, we consider the internals of any fixed protocol $\protSMS$ for the \sms problem. We prove that for \emph{any} protocol \protSMS, there exists 
\emph{some} index $i \in [k]$, such that $\protSMS$ is solving the \bhmz instance encoded by the matching $M_i$ of player $\PS{i}$ and the vector $x_i$, defined by the inputs 
of players $\PS{j}$ for $j \neq i$. In order to prove this, we need to analyze the protocol \protSMS on distributions other than the ones defined above for \sms. 
Interestingly, in these distributions, there is \emph{no} large gap between the matching size (in \Yes and \No cases) and hence a priori it is not even clear why \protSMS should perform any non-trivial task over them. 
Having proved this, we can then embed any instance of \bhmz into an instance of \sms for the \emph{specific} protocol \protSMS, using a careful reduction, in which 
we have to crucially use the fact that \protSMS is a simultaneous protocol (as opposed to one-way) to obtain a one-way protocol for $\bhmz$. 

\paragraph{Dense graphs.} The starting point of our approach in Theorem~\ref{thm:lower-eps-intro} is~\cite{BuryS15} (itself based on a prior result of~\cite{EsfandiariHLMO15}) that establishes 
a reduction for estimating matching size from the \bhhz problem in the one-way model (as mentioned in Section~\ref{sec:prelim}). 
The setup here is as follows: Alice is given a matching $M$, Bob is given a collection of cliques of size $\Theta(1/\eps)$ (denoted by $E_B$) and depending on 
the answer of the ``embedded'' \bhhz problem in the reduction, the maximum matching in $M \cup E_B$ differs by a factor of $(1+\eps)$. This reduction then implies a 
lower bound of $n^{1-O(\eps)}$ by the known lower bounds on the communication complexity of \bhhz~\cite{VerbinY11}. 

To ``boost'' this lower bound from $n^{1-O(\eps)}$ to the \emph{super-linear} regime, a natural idea is to provide Alice not with a single matching $M$, but a collection of $t$ \emph{independently 
chosen} matchings $M_1,\ldots,M_t$, and then ask Alice and Bob to solve the problem for a \emph{uniformly at random} chosen matching $M_{\jstar}$ and
a single collection of $\Theta(1/\eps)$-cliques (provided to Bob as before). Intuitively, Alice now needs to solve $t$ different instances of the \bhhz problem (as index $\jstar$ is not known to Alice) 
and this should make the new problem $t$ \emph{times harder} than the original one. 

There are three main obstacles in implementing this idea: $(i)$ the matchings $M_1,\ldots,M_t$ should be ``supported'' on $\Theta(n)$ vertices, 
as opposed to the trivial $\Theta(t \cdot n)$ vertices (or otherwise the lower bound would be too weak in terms of size of the final graph), $(ii)$ the matchings should be 
chosen independently even though they are supported on essentially the same set of vertices (or otherwise we cannot argue that the new problem is indeed $t$ times harder), 
and finally $(iii)$, the reduction should ensure that Alice and Bob still need to solve the specific embedded \bhhz instance for a uniformly at random chosen
matching $M_\jstar$ and the $\Theta(1/\eps)$-cliques (as otherwise we do not obtain a valid reduction). 

We bypass these obstacles by using RS graphs defined in Section~\ref{sec:prelim}. 
Intuitively, we use RS graphs to ``pack'' the matchings $M_1,\ldots,M_t$ in the aforementioned reduction over $\Theta(n)$ vertices and use the fact that these matchings
are \emph{induced} to ensure the independence between the different matchings. Our reduction can be interpreted as ``embedding'' multiple instances of the \bhhz problem into a single graph. 
RS graphs have been used previously in~\cite{GoelKK12,Kapralov13,AssadiKLY16,Konrad15} for
proving lower bounds for \emph{finding} approximate matchings. 
While it was possible to analyze the hard instances in~\cite{GoelKK12,Kapralov13,AssadiKLY16,Konrad15} using simple counting arguments that crucially exploited the requirement on \emph{outputting 
a valid matching}, we now need to prove the lower bound using 
\emph{information complexity} to reduce the original problem (i.e., matching size estimation) to multiple instances of a simpler problem (i.e., $t$ instances of \bhhz), using 
a direct-sum style argument.  This introduces new challenges, including designing a reduction from a two-player one-way problem like \bhhz to a multi-player simultaneous 
problem that does not ``leak'' much information. En route, we also establish a lower bound on the \emph{information complexity}
of $\bhhz$ that matches the best known lower bound on its communication complexity (see Theorem~\ref{thm:bhh-ic}).  

\subsection{Upper bounds} The main idea behind our algorithms in Theorem~\ref{thm:upper-intro} is the following structural result that we show: if we sample
each \emph{vertex} in a graph $G$ w.p. (essentially) $1/\alpha$, then the maximum matching size in the subgraph $G'$ induced by the sampled vertices
can be used to obtain an $\alpha$-approximate estimate of matching size in $G$.
Using this result, we design an algorithm that samples the vertices of $G$ at a rate $1/\alpha$ to obtain an induced subgraph $G'$, and maintains a sufficiently large matching in $G'$ to estimate $\opt(G)$. 

For insertion-only streams, a large matching (up to $2$ approximation) 
can be computed simply by maintaining a maximal matching in $G'$. For dynamic streams, we use existing
results of~\cite{AssadiKLY16,ChitnisCEHMMV16} that allow finding a (large) matching with size at most $k$ (in our case, $k = \Theta({n / \alpha^2})$) in $\Ot(k^2)$ space in dynamic streams.

\section{An Information Complexity Lower Bound for BHH} \label{sec:bhh}

In this section, we state the known communication complexity results for the \bhh problem defined in Section~\ref{sec:bhh-prelim} and 
then prove an information complexity lower bound for this problem. 


The following is a hard distribution for $\BHH{n}{t}$ used in~\cite{VerbinY11}: 
\textbox{The distribution $\dist$ for $\BHH{n}{t}$.}{
\begin{itemize}
	\item \textbf{Alice}: The input to Alice is a boolean vector $x \in \set{0,1}^{n}$ chosen \emph{uniformly at random}. 
	\item \textbf{Bob}: The input to Bob is a perfect $t$-hypermatching $M$ chosen \emph{uniformly at random} and a boolean vector $w$ such that, w.p. $1/2$, $w = Mx$ and w.p. $1/2$, ${w} = \negative{Mx}$. 
\end{itemize}
}

The result in~\cite{VerbinY11} can now be stated as follows. 

\begin{theorem}[Communication Complexity of $\BHH{n}{t}$~\cite{VerbinY11}]\label{thm:bhh-cc}
	For any $t \geq 2$, suppose $n = 2kt$ for some integer $k \geq 1$, and $\delta \in (0,1/2)$. Let $\gamma := \frac{1}{2} - \delta$; then,
	\[\CCO{\BHH{n}{t}}{\dist}{\delta} = \Omega\paren{\gamma \cdot n^{1-1/t}}\] 
	This bound also holds for the \emph{communication cost} of the protocols that are \emph{only} required to be correct w.p. $1-\delta$ on the distribution $\dist$ (not necessarily on all inputs). 
\end{theorem}

We point out that the communication lower bound for $\BHH{n}{t}$ stated in~\cite{VerbinY11} (and similarly for $\BHM{n}$ stated
in~\cite{GKKRW07}), has a dependence of $\gamma^2$ instead of $\gamma$; however, obtaining the linear dependence on $\gamma$ is
straightforward (see the proof of Theorem~\ref{thm:bhh-ic} in this paper). This dependence is crucial for our results in
Section~\ref{sec:simple-ms} since we are analyzing protocols which outperforms random guessing only by a very small probability.


For our application, we need a (stronger) lower bound on the \emph{information complexity} of $\BHH{n}{t}$ rather than its communication complexity. 
Since this result does not follow directly from those of~\cite{VerbinY11,GKKRW07}, for completeness, we provide a proof of this result in this section
following the approach in~\cite{GKKRW07,VerbinY11}. We remark that one can also use the message compression technique of~\cite{JainRS03} for 
bounded round communication protocols to prove this result. 

\begin{theorem}[Information Complexity of $\BHH{n}{t}$]\label{thm:bhh-ic}
  For any $t \geq 2$, any $n = 2kt$ for some integer $k \geq 1$, and any constant $\delta < 1/2$,
	\[\ICO{\BHH{n}{t}}{\dist}{\delta} = \Omega\paren{n^{1-1/t}}\] 
	This bound also holds for the \emph{information cost} of the protocols that are \emph{only} required to be correct w.p. $1-\delta$ on the distribution $\dist$ (not necessarily on all inputs). 
\end{theorem}

The general idea in the proof of~\cite{GKKRW07,VerbinY11} is as follows: in any protocol $\Prot$, Alice's message partitions the set $\set{0,1}^{n}$ into $2^{c}$ sets $A_1,\ldots,A_{2^{c}}$ (here $c := \norm{\Prot}$) 
and hence for a typical message of Alice, Bob knows that the random variable $\bX$ (of Alice's input) is chosen uniformly at random from some set $A_i$ with $\card{A_i} \geq 2^{n-c}$. 
Now consider the hypermatching $M$ of Bob, and the distributions $M\bX$ and $\negative{M\bX}$ for $\bX$ uniform on $A_i$; the main
technical result in~\cite{GKKRW07,VerbinY11} proves that for any sufficiently large set $A_i$ (size essentially $2^{n-(n^{1-1/t})}$), if $\bX$ is chosen uniformly at random from $A_i$, then 
the distribution of $M\bX$ and $\negative{M\bX}$ look identical to Bob. Consequently, since Bob's task is to, given a vector $w$, decide whether $w$ was chosen from $M\bX$ or $\negative{M\bX}$, the advantage of 
Bob over random guessing would be negligible. 
 
To prove Theorem~\ref{thm:bhh-ic}, we also follow the same approach described above. The main difference here is that to prove an information complexity lower bound we need to work with protocols that are randomized 
\emph{even} on the distribution $\dist$. This makes the problem more challenging since unlike deterministic protocols, randomized protocols do \emph{not} split the input into disjoint distributions that 
are uniform (i.e., the sets $A_1,\ldots,A_{2^c}$ described above). To overcome this, we use Lemma~\ref{lem:aux-high-ent} proved in Section~\ref{sec:info} to partition the inputs
conditioned on Alice's message into several near uniform parts and then apply the aforementioned technical result of~\cite{GKKRW07,VerbinY11} on each part separately to finalize the proof.

We now provide the formal proof. We say that a set $A \subseteq \set{0,1}^{n}$ is a \emph{single parity} set iff the parity (i.e., the $\oplus$ summation) of each string in
$A$ is the same.  The following lemma is the main ingredient of the proof in~\cite{VerbinY11} (see Theorem~3.1 in~\cite{VerbinY11}). 

\begin{lemma}[\!\!\cite{VerbinY11}]\label{lem:bhh-tvd}
  Suppose $n = 2kt$ for some integer $k \geq 1$, $A \subseteq \set{0,1}^{n}$ is a single parity set of size $\card{A} \geq 2^{n-c}$ for some
  $c \geq 1$, and $x$ is a vector drawn uniformly at random from $A$ (denote the distribution by $U_A$). Let $M$ be a perfect
  $t$-hypermatching on $[n]$ chosen uniformly at random; there exists an \emph{absolute constant} $\ell > 0$, such that for all $\eps \in (0,1]$, if $c \leq \ell \cdot \eps \cdot n^{1-1/t}$, then,
	\[
		\EX_{M} \Bracket{\tvd{p_M}{q_M}} \leq \eps
	\]
	where $p_M$ and $q_M$ are distributions over $\set{0,1}^{n/t}$ whose p.d.f are (for any $z \in \set{0,1}^{n/t}$) 
	\[
		p_M(z) := \Pr_{\bX \sim U_A}(M\bX = z) 
	\]
	and
	\[
		q_M(z) := \Pr_{\bX \sim U_A}(M\bX = \negative{z}) 
	\] 
	respectively. In other words, $p_M = M U_A$ and $q_M = \negative{M U_A}$.
\end{lemma}

We are now ready to prove Theorem~\ref{thm:bhh-ic}.



\begin{proof}[Proof of Theorem~\ref{thm:bhh-ic}]
  Define $C_1:= \paren{2\eps^3 \cdot \ell \cdot n^{1-1/t} - 1}$ for a constant $\eps$ (depending on $\delta$) to be determined later\footnote{Here, unlike the case in Theorem~\ref{thm:bhh-cc}, we are \emph{not} particularly interested
  in achieving the best dependence on the parameter $\delta$, since in our proofs that require this theorem we always work with constant values of $\delta$; this allows us to  simplify the proof significantly.}. 
  Suppose towards a contradiction that $\ICO{\BHH{n}{t}}{\dist}{\delta} \leq C_1$ 
 and let $\Prot$ be a $\delta$-error protocol for $\BHH{n}{t}$ on the distribution $\dist$
  with information cost $C_1$. Let $\bProt$ be the random variable denoting the message sent
  from Alice to Bob using $\Prot$, and let $\bR$ be the random variable denoting the public coins used in the protocol. We have,
		\begin{align*}
			\ICO{\BHH{n}{t}}{\dist}{\delta} = I(\bProt;\bX \mid \bR) = H(\bX \mid \bR) - H(\bX \mid \bR,\bProt) = n - H(\bX \mid \bR,\bProt)
		\end{align*}
		where the last equality is because the input $\bX$ is uniform on $\set{0,1}^{n}$ in $\dist$ and is chosen independently of the public coins $\bR$. Consequently, we have
                $H(\bX \mid \bR,\bProt) = n-C_1$.  We further define a random variable $\bP \in \set{0,1}$ that indicates the
                \emph{parity} of the input vector $\bX$. We have,
		\begin{align*}
			H(\bX \mid \bR,\bProt,\bP) \geq H(\bX \mid \bR,\bProt) - H(\bP) = n- C_1 - 1 
		\end{align*}
		where the inequality is by chain rule of entropy (\itfacts{ent-chain-rule}). Hence, 
		\begin{align*}
			&\EX_{R,\prot,P}\Bracket{H(\bX \mid \bR=R,\bProt=\prot,\bP=P)} \ge n-C_1 -1
		\end{align*}
		For brevity, we denote the event $(\bR=R,\bProt=\prot,\bP=P)$ by $Z$. Define $C_2 := \frac{C_1+1}{\eps} = 2\eps^2 \cdot \ell
                \cdot n^{1-1/t}$. By Markov inequality,
		\begin{align*}
			\Pr_{Z}\paren{ n- H(\bX \mid Z) \ge C_2} \leq \frac{\EX_{R,\prot,P}\Bracket{ n- H(\bX \mid \bR=R,\bProt=\prot,\bP=P)} }{C_2} \leq \frac{C_1+1}{C_2} = \eps
		\end{align*}
                Hence, w.p. at least $1-\eps$, $H(\bX \mid Z) \ge n - C_2$. Assuming this event happens, let $\bX_Z$ denote the random
                variable $\bX$ conditioned on $Z$ and hence $H(\bX_Z) \geq n-C_2$.  We emphasize here that the randomness in $\bX_Z$ lies
                both in the distribution $\dist$ and in the \emph{private} randomness used by the protocol.
		
		Now consider the task of Bob for solving $\BHH{n}{t}$. Bob is given a perfect $t$-hypermatching $M$, a vector $w$, a message
                $\bProt = \prot$ (from Alice), and the public coins $\bR=R$. Suppose, additionally, we provide the parity of the input
                $\bX$ to Bob for free , i.e., Bob knows $\bP = P$. Conditioned on the aforementioned event happening (w.p. $1-\eps$), Bob knows that
                the input of Alice is chosen from the distribution of the random variable $\bX_Z$ with $H(\bX_Z) \geq n-C_2$. For the hypermatching $M$ of
                Bob, Bob is given a vector $w$ chosen from the distribution of either $M \bX_Z$ or $\negative{M \bX_Z}$ and his task is to
                distinguish between these two cases. In the rest of the proof, we show that total variation distance of these two
                distributions is small and hence, by Fact~\ref{fact:tvd}, Bob will not able to distinguish them using a single sample.


		Note that since the protocol is \emph{not} deterministic, the distribution of $\bX_Z$ is not necessarily uniform over its
                support and hence we cannot directly apply Lemma~\ref{lem:bhh-tvd} to bound the total variation distance between the
                distributions  $M \bX_Z$ and $\negative{M \bX_Z}$.  To bypass this, we first apply Lemma~\ref{lem:aux-high-ent}
                on the random variable $\bX_Z$ with the parameter $\Delta =  C_2$ 
                and $\eps$  to obtain a sequence of $k+1$ distributions $\mu_0,\ldots,\mu_k$ where $k = O(n/\eps)$.  For any
                $i\geq 1$, let $U_i$ be the uniform distribution over the support of $\mu_i$.  By Lemma~\ref{lem:aux-high-ent},
                $\tvd{\mu_i}{U_i} = O(\eps)$. Since $M$ is chosen independently of $\bX_Z$, this implies that
                $\tvd{M \mu_i}{M U_i} = O(\eps)$ and $\tvd{\negative{M \mu_i}}{\negative{M U_i}} = O(\eps)$. 
                Furthermore, by Lemma~\ref{lem:aux-high-ent}, $\card{\supp{U_i}} = \card{\supp{\mu_i}} \geq 2^{n-\frac{\Delta}{\eps}-\log{\Theta(n/\eps)}} \geq 2^{n-\eps \ell n^{1-1/t}}$, and hence 
         	by Lemma~\ref{lem:bhh-tvd}, $\EX_{M} \Bracket{\tvd{MU_i}{\negative{MU_i}}} = O(\eps)$. Finally, by triangle inequality, 
		$\EX_{M} \Bracket{\tvd{M\mu_i}{\negative{M\mu_i}}}  = O(\eps)$. 
		
		\newcommand{\success}{\ensuremath{\textnormal{\textsf{success}}}\xspace}
		
		Again fix an $i \geq 1$. Suppose we further specify to Bob that $\bX_Z$ is chosen from $\mu_i$.  We say that Bob is
                \emph{successful} if for the given matching $M$ and the vector $w$, he can correctly identify whether $w$ is chosen from
                ${M\mu_i}$ or $\negative{M\mu_i}$; in other words, use one sample (i.e., $w$) to distinguish between ${M\mu_i}$ or
                $\negative{M\mu_i}$. Denote the event that Bob is successful by \success. For the following equations, let the summation
                over $x$ ranges over all possible values of $\tvd{M\mu_i}{\negative{M\mu_i}}$ (since there are only $n!$ matchings, there
                are at most $n!$ choices); we have,\footnote{As stated after Theorem~\ref{thm:bhh-cc}, to obtain a dependence of $\gamma$
                  instead of $\gamma^2$ in the communication complexity lower bound proofs of~\cite{GKKRW07,VerbinY11}, simply replace the
                  Markov bound argument at the end of the their proofs with the slightly more careful argument based on probability of being
                  successful presented here.}
		\begin{align*}
			\Pr(\success) &= \sum_{x} \Pr\paren{\tvd{M\mu_i}{\negative{M\mu_i}} = x} \cdot \Pr\paren{\success \mid \tvd{M\mu_i}{\negative{M\mu_i}} = x }  \\
			&\leq \sum_{x} \Pr\paren{\tvd{M\mu_i}{\negative{M\mu_i}} = x} \cdot \paren{\frac{1}{2}+\frac{x}{2}}  \tag{by Fact~\ref{fact:tvd}} \\
			&= \frac{1}{2} + \frac{1}{2} \cdot \sum_{x \in [0,1]} \Pr\paren{\tvd{M\mu_i}{\negative{M\mu_i}} = x}\cdot x \\
			&= \frac{1}{2} + \frac{1}{2} \cdot \Ex\Bracket{\tvd{M\mu_i}{\negative{M\mu_i}}} = \frac{1}{2} + {O(\eps)}
		\end{align*}
		
		To summarize, the advantage of Bob over randomly guessing the output is at most $\eps$ (for the unlikely event that
                $H(\bX_Z) < n-C_1$) plus $O(\eps)$ (for the unlikely event that $\bX$ is chosen from $\mu_0$ in
                Lemma~\ref{lem:aux-high-ent}) plus $O(\eps)$ (for the advantage over random guessing, i.e., when event \success happens in $\mu_i$
                for $i\geq 1$). In summary, the probability of success of Bob is at most $\frac{1}{2} + O(\eps) < 1-\delta$, by choosing
                $\eps$ small enough in compare to $\delta$. This means that the protocol succeeds w.p. strictly less that $1-\delta$, a
                contradiction.
\end{proof}

For our purpose, it would be more convenient to work with a special case of the $\BHH{n}{t}$ problem, namely $\BHHZ{n}{t}$ in which
the vector $w = 0^{n/t}$ and hence the goal of Bob is simply to decide whether $Mx =0^{n/t}$ (\Yes case) or $Mx =1^{n/t}$ (\No case). We define $\BHMZ{n} := \BHHZ{n}{2}$ (similar to $\BHM{n}$).
It is known that (see,~\cite{VerbinY11,BuryS15,LiW16}) any instance of the original $\BHH{n}{t}$ problem 
can be reduced to an instance of $\BHHZ{2n}{t}$ \emph{deterministically} and with \emph{no communication} between the players. 

The following corollary summarizes the results in this section. 

\begin{corollary}\label{cor:bhh}
	For any $n = 2kt$ (for some integer $k \geq 1$), there exists a distribution $\distbhh$ for $\BHHZ{n}{t}$ such that: 
	\begin{itemize}
		\item For any $\delta \in (0,1)$ and $\gamma:= \frac{1}{2} - \delta$, $\CCO{\BHHZ{n}{t}}{\distbhh}{\delta} = \Omega(\gamma \cdot n^{1-1/t})$.
		\item For any \emph{constant} $\delta < 1/2$, $\ICO{\BHHZ{n}{t}}{\distbhh}{\delta} = \Omega( n^{1-1/t})$.
		\item Alice's input $\bX \sim \distbhh$ is supported on boolean vectors $x \in \set{0,1}^{n}$ with $\norm{x}_0 = \frac{n}{2}$. 
	\end{itemize}
	Moreover, these bounds also hold for, respectively, the communication cost and information cost of the protocols that are \emph{only} required to be correct
  w.p. $1-\delta$ on the distribution $\distbhh$ (not necessarily on all inputs).
\end{corollary} 

We remark that this distribution satisfy the requirement of the Claim~\ref{clm:bhh-matching} (i.e., $\norm{x}_0 = \frac{n}{2}$) and hence 
can be used in the reduction for the matching size problem mentioned in Section~\ref{sec:prelim}.

\section{Space Lower Bounds for $\alpha$-Approximating Matching Size} \label{sec:lb-alpha}

In this section, we present our space lower bounds for $\alpha$-approximation algorithms in dynamic streams. 
As already remarked in Section~\ref{sec:prelim}, by the results of~\cite{AiHLW16,LiNW14}, it suffices to prove the lower bound
in the SMP model. 

\subsection{An $\Omega(\sqrt{n}/\alpha^{2.5})$ Lower Bound for Sparse Graphs} \label{sec:simple-ms}

We consider the sparse graphs case in this section (i.e., Part~(1) of Theorem~\ref{thm:lower-alpha-intro}), and show that any single-pass
streaming algorithm that computes an $\alpha$-approximation of matching size must use $\Omega(\sqrt{n}/\alpha^{2.5})$ bits of space even if
the input graph only have $O(n)$ edges.

Define the \emph{sparse matching size estimation} problem, $\SMS{n}{k}$, as the following $k$-player communication problem in the SMP model:
each player $\PS{i}$ is given a matching $M_i$ over a set $V$ of $n + \frac{n}{k}$ vertices\footnote{To simplify the exposition, we use
  $n+\frac{n}{k}$ instead of the usual $n$ as the number of vertices.} and the goal of the players is to approximate the maximum matching
size of $G(V,\Union_{i \in [k]} M_i)$ to within a factor \emph{better than} $\frac{k+1}{2}$.  We prove the following lower bound on the
communication complexity of $\SMS{n}{k}$.

\begin{theorem}\label{thm:SMS-hard}
  For any sufficiently large $n$, and $k\geq 2$, there exists a distribution $\dist$ for $\SMS{n}{k}$ such that for any constant $\delta < 1/2$:
  $
      \CCS{\SMS{n}{k}}{\dist}{\delta} = \Omega\paren{\frac{\sqrt{n}}{k\sqrt{k}}}
  $
\end{theorem}

Part~(1) of Theorem~\ref{thm:lower-alpha-intro} immediately follows from Theorem~\ref{thm:SMS-hard}.

\begin{proof}[Proof of Theorem~\ref{thm:lower-alpha-intro}, Part~(1)]
  Any SMP protocol for estimating matching size to within a factor of $\alpha < \frac{k+1}{2}$ can be used to solve the $\SMS{n}{k}$ problem.
  Moreover, as stated in Section~\ref{sec:cc-ic}, SMP communication complexity of a $k$-player problem is at most $k$ times the space
  complexity of any single-pass streaming algorithm in dynamic streams~\cite{LiNW14,AiHLW16}; this finalizes the first part of the proof. 

  To see that the space complexity holds even when the input graph is both sparse and having bounded arboricity, notice that any graph $G$
  in $\SMS{n}{k}$ has exactly $k \cdot \frac{n}{k} = n$ edges (hence sparse); furthermore, since each player is given a matching (which is
  always a forest), the arboricity of $G$ is at most $k \leq 2\alpha$.
\end{proof}
 
In the following, we focus on proving Theorem~\ref{thm:SMS-hard}.  This theorem is ultimately proved by a reduction from the
\bhmz problem defined in Section~\ref{sec:bhh}. However, this reduction is non-standard in the sense that it is \emph{protocol-dependent}:
given any protocol $\Prot$ for \sms, we create a protocol for \bhmz by \emph{embedding} an instance of \bhmz in the input of \sms, whereby
the embedding is designed specifically for the protocol $\Prot$. It is worth mentioning that \bhmz is a hard problem even in the one-way
model, while the distribution that we create for $\sms$ is only hard in the SMP model, meaning that if any player is allowed to send a
single message to any other player (instead of the referee), then $\Ot(1)$ bits of communication suffices to solve the problem. Therefore, a
key technical challenge here is to design a reduction from a one-way problem to a problem that is ``inherently'' simultaneous, or in other
words, is easy to solve in the one-way model.

\subsubsection{A Hard Input Distribution for $\SMS{n}{k}$}\label{sec:sms-dist}
Let $\DistBHM$ be the hard input distribution of $\BHMZ{\frac{2n}{k}}$ in Corollary~\ref{cor:bhh} (for $t=2$) and $\DistBHMY$ and $\DistBHMN$
be, respectively, the distribution on \Yes and \No instances of $\DistBHM$.  

\textbox{The distribution $\DistSMS$ for $\SMS{n}{k}$:}{
  \begin{enumerate}
 
  \item For each $i \in [k]$, independently draw a $\BHMZ{{\frac{2n}{k}}}$ instance $(\Mbhm_i,\xbhm_i) \sim \DistBHM$.
    
    \item Draw a \emph{random} permutation $\sigma: \Bracket{n + {n\over k}} \rightarrow \Bracket{n + {n\over k}}$. 
    
    \item For each player $i \in [k]$, we define a mapping $\sigma_i: [{2n \over k}] \rightarrow \bracket{n + {n\over k}}$ as follows:  
    \begin{itemize}
    	\item For each  $j \in [{2n \over k}]$ with $\xbhm_i(j) = 1$, if $\xbhm_i(j)$ is the $\ell$-th smallest index
      	with value $1$, let $\sigma_i(j):= \sigma(\ell)$\footnote{Here, we use the fact that $\norm{\xbhm_i}_0 = \frac{n}{k}$ in $\DistBHM$ by Corollary~\ref{cor:bhh}}.
	\item For each  $j \in [{2n \over k}]$ with $\xbhm_i(j) = 0$, if $\xbhm_i(j)$ is the $\ell$-th smallest index
      	with value $0$, let $\sigma_i(j) := \sigma(i \cdot {n\over k} + \ell)$.
    \end{itemize}
    \item The input to each player $\PS{i}$ is a matching $M_i := \set{\paren{\sigma_i(u),\sigma_i(v)} \mid (u,v) \in \Mbhm_i}$. 

  \end{enumerate}
}

Observe that the distribution $\DistSMS$ is defined by $k$ instances of $\BHMZ{{\frac{2n}{k}}}$, i.e., $(\Mbhm_i, \xbhm_i)$ (for $i \in [k]$), along with a
mapping $\sigma$. The mapping $\sigma$ relates the vectors $\xbhm_i$ to the set of vertices in the final graph $G$ while ensuring that
across the players, for any $j \in [{2n \over k}]$ where $\xbhm_i(j)=1$, the vertex that $j$ maps to is \emph{shared}, while the vertices
with $\xbhm_{i}(j)=0$ are \emph{unique} to each player. Moreover, the mapping $\sigma_i$ provided to each player effectively describes the
set of vertices (denoted by $V_i$) that the edges of $\PS{i}$ will be incident on, and the matching $\Mbhm_i$ describes the edges between
$V_i$. Hence, we can \emph{uniquely} define the input of each player $\PS{i}$ by the pair $(\Mbhm_i,\sigma_i)$, and from now on, without
loss of generality, we assume the input given to each player $\PS{i}$ is the pair $(\Mbhm_i,\sigma_i)$. 

We should note right away that the distribution $\DistSMS$ is not a ``hard'' distribution for $\SMS{n}{k}$ in the traditional sense: it is not
hard to verify that for any graph $G \sim \DistSMS$, $\opt(G)$ is concentrated around its expectation, and hence it is trivial to design a
protocol when instances are {promised} to be \emph{only} sampled from $\DistSMS$: always output $\Ex_{G \sim \DistSMS}\bracket{\opt(G)}$, which requires no communication from
the players.

Nevertheless, the way we use the distribution $\DistSMS$ as a hard distribution is to consider any protocol $\protSMS$ that succeeds \emph{uniformly}, i.e., on \emph{any} instance of
$\SMS{n}{k}$; we then execute $\protSMS$ on $\DistSMS$ and argue that in order to perform
well on every instance of $\DistSMS$, $\protSMS$ must convey a non-trivial amount of information about the input of the players in \emph{some sub-distribution} of $\DistSMS$.  To
continue, we need the following definitions.

\begin{definition}[Input Profile]
  For each graph $G \sim \DistSMS$, we define the input profile of $G$ to be a vector $f \in \set{\Yes, \No}^{k}$, where
  $f(i) = \Yes$ iff the $i$-th \bhm instance $(\Mbhm_i, \xbhm_i)$ in $G$ is a \Yes  instance and otherwise $f(i) = \No$.
\end{definition}

The $2^{k}$ different possible input profiles partition $\DistSMS$ into $2^{k}$ different distributions. 
For any input profile $f$, we use the notation $\DistSMS \mid f$ to denote the distribution of $\DistSMS$ \emph{conditioned} on its input profile 
being $f$. Two particularly interesting profiles for our purpose are the \emph{all-equal} profiles, i.e., $\fYes:= (\Yes,\ldots,\Yes)$ and $\fNo := (\No,\ldots,\No)$, due to the 
following claim. 

\begin{claim}\label{clm:SMS-size-gap}
  For any graph $G \sim \paren{\DistSMS \mid \fYes}$, $\opt(G) \geq {n \over 2} + \frac{n}{2k}$, and for any graph $G \sim \paren{\DistSMS \mid \fNo}$, $\opt(G) \leq {n \over k}$. 
\end{claim}
\begin{proof}
  In $(\DistSMS \mid \fYes)$, each \bhm instance $(\Mbhm_i, \xbhm_i)$ (for $i \in [k]$) is drawn from $\DistSMSY$, meaning that for every edge
  $(u,v) \in \Mbhm_i$, $\xbhm_i(u) \oplus \xbhm_i(v) = 0$. Therefore, either $\xbhm_i(u) = \xbhm_i(v) = 0$ or $\xbhm_i(u) = \xbhm_i(v) = 1$. 
  Since $\Mbhm_i$ is a perfect matching over the set $[{2n \over k}]$ and the hamming weight of $\xbhm_i$ is ${n \over k}$ (by Corollary~\ref{cor:bhh}), for half of the
  edges in $\Mbhm_i$, we must have $\xbhm_i(u) = \xbhm_i(v) = 0$. Moreover, as $\DistSMS$ maps every vertex with $\xbhm_i(j) = 0$ to a distinct vertex in
  $G$, these $\half \cdot \card{\Mbhm_i} = {n \over 2k}$ edges are vertex-disjoint with any other edge in the final graph $G$. Hence, between the $k$ players, these edges together
  form a matching of size $k \cdot {n \over 2k} = {n \over 2}$. Finally, there is also a 
  matching of size $\frac{n}{2k}$ between the shared vertices: simply use the edges corresponding to a matching $\Mbhm_i$ of an arbitrary player $\PS{i}$ that are incident 
  on shared vertices. This means that in this case, $\opt(G) \geq {n \over 2} + \frac{n}{2k}$. 

  In $(\DistSMS \mid \fNo)$, each $\bhm$ instance $(\Mbhm_i, \xbhm_i)$ (for $i \in [k]$) is drawn from $\DistSMSN$, meaning that for every edge
  $(u,v) \in \Mbhm_i$, $\xbhm_i(u) \oplus \xbhm_i(v) = 1$. Therefore, exactly one of $\xbhm_i(u)$ or $\xbhm_i(v)$ is equal to $1$. In $\DistSMS$, for every
  player, the vertices where $\xbhm_i(j) = 1$ are all mapped to the (same) set of vertices $\set{\sigma(1), \sigma(2), \ldots, \sigma({n \over k})}$
  (denoted by $V_0$). Therefore, in the final graph $G$, \emph{every} edge of \emph{every} player is incident on some vertex in $V_0$, and
  hence the maximum matching size in $G$ is at most $\card{V_0} = {n \over k}$.
\end{proof}

In the following, we fix any $\delta$-error protocol $\protSMS$ for $\SMS{n}{k}$. By Claim~\ref{clm:SMS-size-gap}, $\protSMS$ is
also a $\delta$-error protocol for distinguishing between the two distributions $\paren{\DistSMS \mid \fYes}$ and
$\paren{\DistSMS \mid \fNo}$: simply output $\Yes$ if the estimate of $\opt(G)$ is strictly larger than ${n \over k}$ and output $\No$
otherwise.  From here on, with a slight abuse of notation, we say that $\protSMS$ outputs \Yes whenever it estimates $\opt(G)$ strictly
larger than ${n \over k}$ and outputs \No otherwise (this notation is defined over \emph{any} input, not necessarily chosen from
$\paren{\DistSMS \mid \fYes}$ or $\paren{\DistSMS \mid \fNo}$).

Intuitively, to distinguish between $\paren{\DistSMS \mid \fYes}$ and $\paren{\DistSMS \mid \fNo}$, one should solve (at least one of) the
\bhmz instances embedded in the distribution. This naturally suggests the possibility of performing a reduction from \bhmz and arguing that 
the distribution on $\paren{\DistSMS \mid \fYes}$ and $\paren{\DistSMS \mid \fNo}$ is a hard distribution for $\SMS{n}{k}$. However, in the
case of these two distributions, the $k$ $\bhmz$ instances are highly correlated and hence it is hard
to reason about which \bhmz instance is ``actually being solved''. To get around this, we try $\protSMS$ on other input profiles, with,
informally speaking, less correlation across the $\bhm$ instances.  An immediate issue here is that, unlike the case for the distributions
$\paren{\DistSMS \mid \fYes}$ and $\paren{\DistSMS \mid \fNo}$, the matching sizes for graphs drawn from the other input profiles do not
have a large gap.  Hence, a priori it is not even clear what the actual task of $\protSMS$ is, or why $\protSMS$ should be able to distinguish them. 
However, we show that there are special pairs of input profiles (other than $\fYes$ and $\fNo$) with our desired property (i.e., ``low'' correlation between the \bhmz instances) 
that $\protSMS$ is still able to distinguish. These pairs are ultimately connected to the (property
of) protocol $\protSMS$ itself and hence vary across different choices for the protocol $\protSMS$; this is the main reason that we perform a
protocol-dependent reduction in our proof.

For any input profile $f$, define $\pfYes{f}$ (resp. $\pfNo{f}$) as the probability that $\protSMS$ outputs \Yes (resp. \No) when its input
is sampled from $\DistSMS \mid f$.  We define the notation of \emph{informative index} for the protocol $\protSMS$.

\begin{definition}[Informative Index]
  We say that an index $i \in [k]$ is \emph{$\gamma$-informative} for the protocol $\protSMS$ iff there exist two input profiles $f$ and
  $g$ where $f(i) = \Yes$, $g(i) = \No$, and $f(j) = g(j)$ for all $j \neq i$, such that
  $\pfYes{f} + \pfNo{g} \geq 1+2\gamma$. In this case, the input profiles $f$ and $g$ are called the \emph{witness} of $i$.
\end{definition}

Informally speaking, if $\protSMS$ has a $\gamma$-informative index $i$, then $\protSMS$ can distinguish whether the $i$-th \bhmz instance is
a $\Yes$ or $\No$ instance w.p. at least $\half + \gamma$ (i.e., $\protSMS$ solves the $i$-th \bhmz instance).  In the rest of this section, 
we prove that indeed every protocol $\protSMS$ has an informative index.

\begin{lemma}\label{lem:informative-index}
Any $\delta$-error protocol $\protSMS$ for $\sms$ has a $\gamma$-informative index for $\gamma = \frac{1-2\delta}{2k}$.
\end{lemma}

\begin{proof}
  Suppose towards a contradiction that for any two input profiles $f$ and $g$ that differ only on one entry (say $i$, and $f(i) = \Yes$, $g(i) = \No$), 
  we have, $\pfYes{f} + \pfNo{g} < 1+2\gamma$ for $\gamma = \frac{1 - 2\delta}{2k}$. 

  Consider the following sequence of $(k + 1)$ input profiles:
  \begin{align*}
  (\fYes=) (\Yes,\Yes, \ldots, \Yes ),     (\No,\Yes, \ldots, \Yes ),   (\No,\No, \ldots, \Yes ), \ldots,   (\No,\No, \ldots, \No) (= \fNo)
  \end{align*}
  whereby, for the $j$-th input profile of this sequence (denoted by $f_j$), the first $j-1$ entries of $f_j$ are all \No, and the
  rest are all \Yes. 
  
  Observe that for any $j \in [k]$, the input profiles $f_j$ and $f_{j+1}$ differ in exactly one entry $j$, and $f_j(j) = \Yes$,
  while $f_{j+1}(j) = \No$. Hence, by our assumption, we have $\pfYes{f_j} + \pfNo{f_{j+1}} < 1 + 2\gamma$, which implies
  \begin{align*}
    \pfYes{f_j} < 1 + 2\gamma - \pfNo{f_{j+1}} = \pfYes{f_{j+1}} + 2\gamma \tag{$\pfYes{f_{j+1}} + \pfNo{f_{j+1}} = 1$}
  \end{align*}
  Therefore, 
  \begin{align*}
    \pfYes{f_1} &< \pfYes{f_{2}} + 2\gamma< \pfYes{f_{3}} + 2\gamma \cdot 2 < \dots < \pfYes{f_{k+1}} + 2\gamma \cdot k    
  \end{align*}
  which implies (by adding $\pfNo{f_{k+1}}$ to both sides of the inequality)
  \begin{align}
    \pfYes{f_1} + \pfNo{f_{k+1}}  <   \pfYes{f_{k+1}} + \pfNo{f_{k+1}} + 2\gamma \cdot k  =  1  + 2\gamma \cdot k = 2 \cdot (1-\delta) \label{eq:f-cont}
  \end{align}
  by our choice of $\gamma$. However, since $\protSMS$ is a $\delta$-error protocol for $\SMS{n}{k}$, by Claim~\ref{clm:SMS-size-gap}, the probability that $\PiSMS$ succeeds
  in distinguishing $\paren{\DistSMS \mid \fYes}$ from $\paren{\DistSMS \mid \fNo}$ on the
  distribution 
  $\frac{1}{2} \paren{\DistSMS \mid \fYes} + \frac{1}{2} \paren{\DistSMS \mid \fNo}$ is
  at least $1-\delta$. Therefore,
  $\half \cdot (\pfYes{f_1} + \pfNo{f_{k+1}}) \ge 1 -\delta$, a contradiction to Eq~(\ref{eq:f-cont}). 
\end{proof}

In the next section, we use existence of a $\gamma$-informative index in any protocol $\protSMS$ for $\SMS{n}{k}$ to obtain a protocol for $\BHMZ{\frac{2n}{k}}$ w.p. of success at least $\frac{1}{2} + \gamma$, based 
$\protSMS$. 

\subsubsection{The Reduction From the \BHMZ{\frac{2n}{k}} Problem}\label{sec:sms-reduction}

Recall that $\protSMS$ is a $\delta$-error protocol for the distribution $\DistSMS$. Let $\istar$ be a
$\gamma$-informative index of $\protSMS$ (as in Lemma~\ref{lem:informative-index}), and let input profiles $\fistar$ and $\gistar$ be the witness of $\istar$. 

\textbox{Protocol \protBHM. \textnormal{A protocol for reducing \BHMZ{\frac{2n}{k}} to \SMS{n}{k}}}{\medskip \\
  \textbf{Input:} An instance $(M,x) \sim \DistBHM$ of $\BHMZ{{2n \over k}}$. \\
  \textbf{Output:} \Yes if $Mx=0^{\frac{n}{k}}$ and \No if
  $Mx=1^{\frac{n}{k}}$.  \\ \algline
  \begin{enumerate}
\item Bob creates the input $(\Mbhm_{\istar},\sigma_{\istar})$ for the player $\PS{\istar}$ as follows:  
 \begin{itemize}
 \item Let $\Mbhm_{\istar} = M$.
 \item Using \emph{public randomness}, Bob picks $\sigma_{\istar}$ to be a \emph{uniformly random} injection from $[{2n \over k}]$ to
   $[n + {n \over k}]$.
 \item Let $V_{\istar}$ be the image of $\sigma_{\istar}$ (i.e., $V_{\istar} = \set{\sigma_{\istar}(j) \mid j\in[{2n \over k}]}$).

  \end{itemize}
\item Alice generates the inputs for all other players. Using \emph{private randomness}, Alice first randomly partitions the set
  $[n+\frac{n}{k}] \setminus V_{\istar}$ into $(k - 1)$ sets $\set{V'_i}_{i \in [k] \setminus \set{\istar}}$, where each
  $V'_{i}$ has size ${n \over k}$.  She then generates the input of each player \PS{i} ($i \neq \istar$) as
  follows:
  \begin{itemize}
  \item If $\fistar(i) = \Yes$ (resp. $\fistar(i)=\No$), Alice draws a \BHMZ{\frac{2n}{k}} instance $(\Mbhm_i,\xbhm_i)$ from $\DistBHMY$
    (resp. from $\DistBHMN$). 
  \item The mapping $\sigma_i: [{2n \over k}] \rightarrow [n + {n \over k}]$ is defined as follows. For the ${n \over k}$ entries in
    $[{2n \over k}]$ where $x_i$ is $0$, Alice assigns a \emph{uniformly random} bijection to $V'_i$. For each entry $j$ in $[{2n \over k}]$
    where $\xbhm_i(j) = 1$, suppose $\xbhm_i(j)$ is the $\ell$-th $1$ of $x_i$, Alice assigns $\sigma_i(j) = \sigma_{\istar}(j')$ where $j'$ is the
    index such that $x(j')$ is the $\ell$-th $1$ of $x$.\footnote{Recall that $x$ is the input  vector to Alice in a \bhmz
      instance.}
  \end{itemize}
  \item Bob runs $\protSMS$ for the $\istar$-th player and Alice runs $\protSMS$ for all other players and sends the messages of all other
  players to Bob.
  \item After receiving the messages from Alice, Bob runs the referee part of the protocol $\protSMS$, and outputs the same answer as $\protSMS$. 
\end{enumerate}
}

It is relatively straightforward to verify that the distribution of the instances created by this reduction and the original distributions $(\DistSMS \mid \fistar)$ and $(\DistSMS \mid \gistar)$ 
are identical. Formally,  

\begin{claim}\label{clm:sms-dist-match}
  Suppose $(M,x)$ is a \Yes (resp. \No) \bhm instance; then the \sms instance constructed by Alice and Bob in the given reduction is 
  sampled from $\DistSMS\mid \fistar$ (resp. $\DistSMS \mid \gistar $).
\end{claim}

The proof of this claim is deferred to the end of this section. 

\begin{proof}[Proof of Theorem~\ref{thm:SMS-hard}]
  
  Let $\gamma = \frac{1-2\delta}{2k}$; we first argue that $\protBHM$ outputs a correct answer for \BHMZ{\frac{2n}{k}} w.p. at least $\frac{1}{2} + \gamma$.  
  If the input \bhmz instance $(M,x)$ is a \Yes (resp. \No) instance, then by Claim~\ref{clm:sms-dist-match}, the distribution of the \sms instance
  created in $\protBHM$ is exactly $\DistSMS \mid \fistar$ (resp. $\DistSMS \mid \gistar $); consequently, $\protSMS$ outputs the correct answer w.p. 
  $\frac{1}{2} \cdot \paren{\pfYes{\fistar} + \pfNo{\gistar}}$. Since $\istar$ is a $\frac{1-2\delta}{2k}$-informative instance, 
  we have $\frac{1}{2} \cdot \paren{\pfYes{\fistar} + \pfNo{\gistar}} \geq \frac{1}{2} + \frac{1-2\delta}{2k} = \frac{1}{2} + \gamma$ and hence the protocol $\PiBHM$ 
  outputs the correct answer w.p. at least $\half + \gamma$. 
  
  Now notice that in $\protBHM$, Alice is sending messages of $k-1$ players in $\protSMS$ to Bob and hence communication cost of $\protBHM$ is at most the communication cost of $\protSMS$. 
  Since solving $\BHM{\frac{2n}{k}}$ on $\DistBHM$ w.p. of success $\frac{1}{2} + \gamma$ requires at least $\Omega(\gamma \cdot \sqrt{\frac{n}{k}})$ bits of communication by Corollary~\ref{cor:bhh}, we have
  $\norm{\protSMS} = \Omega( \gamma \cdot \sqrt{\frac{n}{k}})$. Moreover, $\gamma = \frac{\eps}{k}$ for 
  some constant $\eps$ bounded away from $0$ (since $\delta$ is a constant bounded away from $1/2$), hence we obtain that
  $      \CCS{\SMS{n}{k}}{\dist}{\delta} = \Omega\paren{\frac{\sqrt{n}}{k\sqrt{k}}} $
  for $\dist := \frac{1}{2} \paren{\DistSMS \mid \fistar} + \frac{1}{2} \paren{\DistSMS \mid \gistar}$. 
\end{proof}

It only remains to prove Claim~\ref{clm:sms-dist-match}.

\begin{proof}[Proof of Claim~\ref{clm:sms-dist-match}]
  Suppose the input \bhmz instance $(M,x)$ is a \Yes instance. We need to prove that the distribution of the \sms instance created by the
  protocol $\PiBHM$ is the same as the distribution $\DistSMS \mid \fistar$ (the case where $(M,x)$ is a \No instance is similar).  In the
  following, we will go through the construction of $\DistSMS$ conditioned on $\fistar$ step by step and explain how the reduction captures each
  step of the construction.

  Firstly, in the distribution $\DistSMS \mid \fistar$, we draw $k$ instances of \bhmz, where the $i$-th instance $(\Mbhm_i, \xbhm_i)$ is drawn from the \Yes
  (resp. \No) \bhmz instances if $\fistar(i) = \Yes$ (resp. $\fistar(i) = \No$). It is straightforward to verify that for any $i \neq \istar$, in the reduction, the \bhmz instance created by Alice 
  is drawn following $\fistar(i)$. The $\istar$-th instance corresponds to the original input of Alice and Bob and since we assume it is a \Yes instance, this instance is also sampled following 
  $\fistar(i)$. This implies that the first part of the input of every player (i.e., the matchings over $[{2n \over k}]$) are drawn the way in the reduction as in the original distribution. 

  Secondly, in the original distribution, we draw a random permutation $\sigma: [n+ {n \over k}] \rightarrow [n+ {n \over k}]$ and use $\sigma$ to
  define the mapping $\sigma_i$ of each player $i$: map the vertices $j$ ($j \in [{2n \over k}]$) where $x_i(j) = 1$ (resp. $x_i(j) = 0$) to
  the same set of vertices $\set{\sigma(1), \sigma(2), \ldots, \sigma({n \over k})}$ (resp.  to a private set of vertices
  $\set{\sigma(i \cdot {n \over k} + 1), \sigma(i \cdot {n \over k} + 2), \ldots, \sigma(i \cdot {n \over k} + {n \over k})}$) in the final graph $G$.

  In the reduction, Bob picks a random injection $\sigma_{\istar}$, and defines the image of $\sigma_{\istar}$ by $V_{\istar}$; Alice
  randomly partition $[n + {n \over k}] \setminus V_{\istar}$ into $k-1$ sets of size ${n \over k}$ each.  For each of the other player $i \neq \istar$,
  Alice assigns the entries in $[{2n \over k}]$ where $x_i$ is $1$ to $V^*$ following the same procedure as $\DistSMS$. For the entries in
  $[{2n \over k}]$ where $x_i$ is $0$, Alice assigns a random bijection to $V_i$.

  To see that the two ways of defining the mappings $\sigma_i$'s are equivalent in the original distribution and in the reduction, simply note that one
  can decompose the choice of the random permutation $\sigma$ into the following steps:
  \begin{enumerate}
  \item Pick a random subset of size ${2n \over k}$ (i.e., $V_{\istar}$). 
  \item Partition the remaining universe into $k-1$ sets of size ${n \over k}$ each.
  \item Pick a random permutation for each set (which is equivalent to picking a random bijection as used in the reduction).
  \end{enumerate}

  Therefore, the mappings $\sigma_i$'s created by the reduction induce the same distribution as the mappings created by $\DistSMS$. Hence, the
  distribution of the inputs of all players created by the protocol $\protBHM$ is the same as the distribution $\DistSMS \mid f_1$.
\end{proof}


\subsection{An $\Omega(n/\alpha^2)$ Lower Bound for Dense Graphs}\label{sec:dense-matching}

We switch to the dense graphs case in this section (i.e., Part~(2) of Theorem~\ref{thm:lower-alpha-intro}), and establish a better lower
bound of $\Omega(n/\alpha^2)$ for computing an $\alpha$-approximate matching size in dynamic streams.

We define \MSa{n,k,\alpha} as the \emph{$k$-player simultaneous communication} problem of estimating the matching size to within a factor of
$\alpha$, when edges of an $n$-vertex input graph $G(V,E)$ are partitioned across the $k$-players.  In this section, we prove the following
lower bound on the information complexity of $\MSa{n,k,\alpha}$ in the SMP communication model.

\begin{theorem}[Lower bound for \MS{n,k,\alpha}]\label{thm:MSa-hard}
  For any sufficiently large $n$ and $\alpha$, there exists some $k = \alpha \cdot \paren{\frac{n}{\alpha}}^{o(1)}$ and a distribution $\distMSa$
  for $\MSa{n,k,\alpha}$ such that for any constant $\delta < \frac{1}{2}$:
  \[
      \ICS{\MSa{n,k,\alpha}}{\distMS}{\delta} = \Omega(nk/\alpha^2)
  \]
\end{theorem}

Theorem~\ref{thm:MSa-hard}, combined with Proposition~\ref{prop:cc-ic}, immediately gives the same lower bound on the SMP communication
complexity of $\MSa{n,k,\alpha}$. Since SMP communication complexity of a $k$-player problem is at most $k$ times the space complexity of
any single-pass streaming algorithm in dynamic streams~\cite{LiNW14,AiHLW16}, this immediately proves the $\Omega(n/\alpha^2)$ lower bound in Part~(2) of
Theorem~\ref{thm:lower-eps-intro}. We now prove Theorem~\ref{thm:MSa-hard}.


\textbox{The hard distribution $\distMSa$ for $\MSa{n,k,\alpha}$:}{
\begin{enumerate}
\item[] \textbf{Parameters:} $r = N^{1-o(1)},~t = \frac{{N \choose 2} - o(N^2)}{r},~k = \frac{(\alpha+1) N}{r},~n=N + k\cdot
  r$.
\item Fix an $(r,t)$-RS graph $\Grs$ on $N$ vertices with induced matchings $\Mrs_1,\ldots,\Mrs_t$.
\item Pick $\jstar \in [t]$ and $\theta \in \set{0,1}$ independently and uniformly at random. 

\item For each player \PS{i}  independently,
  \begin{enumerate}
  \item Denote by $G_i$ the input graph of \PS{i}, initialized to be a copy of $\Grs$ with vertices $V_i = [N]$. Moreover, define $\vstari$ as the set of vertices incident 
  on the matching $\Mrs_{\jstar}$. 
  \item Let $x^{(i)}$ be a $t$-dimensional vector, whereby $x^{(i)}(\jstar) = \theta$ and for any $j \neq \jstar$, $x^{(i)}(j)$ is chosen uniformly at random from $\set{0,1}$.  
  \item For any $j \in [t]$, if $x^{(i)}(j) = 0$ \emph{remove} the matching $\Mrs_\jstar$ from $G_i$ (otherwise, do nothing).
   \end{enumerate}
\item Pick a random permutation $\sigma$ of $[n]$. For every player $\PS{i}$, for each vertex $v$ in $V_i\setminus \vstari$ with label $j$
  ($\in [N]$), \emph{relabel} $v$ to $\sigma(j)$. Enumerate the vertices in $\vstari$ (from the one with the smallest label to the largest),
  and relabel the $j$-th vertex to $\sigma(N + (i-1) \cdot 2r + j)$.  In the final graph, the vertices with the same label correspond to
  the same vertex.
\end{enumerate}
} 

The vertices whose labels belong to $\sigma([N])$ are referred to as \emph{shared} vertices since they belong to the input graph of \emph{every}
player, and the vertices $\vstari$ are referred to as the \emph{private} vertices of the player $\PS{i}$ since they only appear in
the input graph of $\PS{i}$ (in the final graph, i.e., after relabeling). We point out that, in general, the final graph constructed by this distribution is a 
multi-graph with $n$ vertices and $O(kN^2) = O(n^{2}/\alpha)$ edges (counting the multiplicities); the multiplicity of each edge is also at most $k$. 
Finally, the existence of an $(r,t)$-RS graph $\Grs$ with the parameters used in this distribution is guaranteed by a result of~\cite{AlonMS12} (see Section~\ref{sec:rs-graphs}).

\begin{claim}\label{clm:alpha-matching}
  Let:
  \begin{align*}
  	\opt_{1} &:= \min_{G}\Paren{\opt(G) \mid \text{$G$ is chosen from $\distMS$ \emph{conditioned} on $\theta = 1$}} \\
	\opt_{0} &:= \max_{G}\Paren{\opt(G) \mid \text{$G$ is chosen from $\distMS$ \emph{conditioned} on $\theta = 0$}}.
  \end{align*}
  then, $\mu_{1} > \alpha \cdot \mu_{0}$.  
\end{claim}

\begin{proof}
  Notice that in each graph $G_i$, except for the matching $\Mrs_{\jstar}$, all other matching edges are incident on the set of shared
  vertices. This implies that across the players, the total contribution of all matchings except for $\Mrs_{\jstar}$'s is at most
  $N$. Consequently, when $\theta = 0$, i.e., when the matching $\Mrs_{\jstar}$ of each player is removed, $\opt(G) \leq N$. On the other
  hand, when $\theta = 1$, since the matching $\Mrs_{\jstar}$ of each player is incident on a unique set of vertices of $G$ (i.e., private
  vertices), they form a matching of size $k \cdot r = (\alpha+1) \cdot N$. Hence, $\opt(G) \geq (\alpha+1) \cdot N$ in this case. 
\end{proof}

Claim~\ref{clm:alpha-matching} shows that any $\delta$-error protocol $\protMSa$ for $\MSa{n,k,\alpha}$ can determine the value of the
parameter $\theta$ in the distribution $\distMSa$ (also with error prob. $\delta$). We use this fact to prove a lower bound on the mutual information
between the parameter $\theta$ and the message of the players. Define $\bT$, $\bsigma$, and $\bJ$, as random variables for, respectively, the parameter $\theta$, the random permutation $\sigma$, and the 
index $\jstar$ in the distribution. We have the following simple claim. 

\begin{claim}\label{clm:alpha-theta}
	For any $\delta < 1/2$ and $\delta$-error protocol $\protMSa$, $I(\bT; \bProtMSa \mid \bsigma,\bJ,\bR) = \Omega(1)$. 
\end{claim}
\begin{proof}
   As proven in Claim~\ref{clm:alpha-matching}, protocol $\protMSa$ can be used directly to determine the value of $\bT$ w.p. $1-\delta$. 
   Hence, by Fano's inequality (Claim~\ref{clm:fano}), $H(\bT \mid \bProtMSa,\bR) \leq H_2(\delta)$, since $\protMSa$ uses the message $\bProtMSa$ together with the public coins $\bR$ to output
  the answer. We further have, 
  \begin{align*}
    H_2(\delta) &\geq H(\bT \mid \bProtMSa,\bR) \geq H(\bT \mid \bProtMSa,\bR,\bSigma,\bJ) \tag{conditioning reduces the entropy (\itfacts{cond-reduce})} \\
                &= H(\bT \mid \bSigma,\bR,\bJ) - I(\bT; \bProtMSa \mid \bR,\bSigma,\bJ) = 1- I(\bT; \bProtMSa \mid \bR,\bSigma,\bJ)
  \end{align*}
  where the last equality is because $\bT$ is chosen uniformly at random from $\set{0,1}$ independent of $\bsigma,\bR$ and
  $\bJ$. Consequently, we have $I(\bT; \bProtMSa \mid \bR,\bSigma,\bJ) \geq 1-H_2(\delta) = \Omega(1)$ (since $\delta < 1/2$ is a constant). 
\end{proof}

We now use the bound in Claim~\ref{clm:alpha-theta} to lower bound the information cost of any $\delta$-error protocol $\protMSa$. 

\begin{lemma}\label{lem:alpha-info-cost}
	For any $\delta < 1/2$ and $\delta$-error protocol $\protMSa$, $\ICost{\protMSa}{\distMSa} = \Omega(t)$. 
\end{lemma}  
\begin{proof}
By Claim~\ref{clm:alpha-theta}, $I(\bT; \bProtMSa \mid \bR,\bSigma,\bJ) = \Omega(1)$; 
in the following, we prove that for this to happen, information cost of $\protMSa$ needs to be $\Omega(t)$.
We have, 
  \begin{align*}
    I(\bT; \bProtMSa \mid \bR,\bSigma,\bJ) &= \Ex_{j \in [t]} I(\bT ; \bProtMSa \mid \bR,\bSigma,\bJ = j) \\
                                           &= \frac{1}{t} \sum_{j=1}^{t} I(\bT ; \bProtMS \mid \bSigma,\bR,\bJ=j) \\
                                           &= \frac{1}{t} \sum_{j=1}^{t} I(\bY_j ; \bProtMS^{(1)},\ldots,\bProtMS^{(k)} \mid \bSigma,\bR,\bJ=j)
  \end{align*}
  Here, for any $j \in [t]$, $\bY_j := (\bX_{1,j},\bX_{2,j},\ldots,\bX_{k,j})$, where $\bX_{i,j}$ (for any $i \in [k]$) is the random
  variable denoting $x^{(i)}(j)$. This equality holds since conditioned on $\bJ = j$, for any $i \in [k]$, each $x^{(i)}(j)$ is assigned to
  be $\theta$ (i.e., $\bT = \bX_{i,j}$ conditioned on $\bJ = j$) . Moreover, 
	\begin{align*}
		I(\bT; \bProtMSa \mid \bR,\bSigma,\bJ) &\leq \frac{1}{t} \sum_{j=1}^{t} \sum_{i=1}^{k} I(\bY_j ; \bProtMS^{(i)} \mid \bSigma,\bR,\bJ=j)
	\end{align*}
	by conditional sub-additivity of mutual information (\itfacts{info-sub-additivity}) since for any $i \in [k]$, $\bProtMS^{(i)}$ and $\bProtMS^{-i}$ are \emph{independent} conditioned on $\bY_j,\bSigma,\bR$ and $\bJ=j$. 
	We can also drop the conditioning on the event $\bJ = j$ and have,
	\begin{align*}
	I(\bT; \bProtMSa \mid \bR,\bSigma,\bJ) &\leq \frac{1}{t} \sum_{j=1}^{t} \sum_{i=1}^{k} I(\bY_j ; \bProtMS^{(i)} \mid \bSigma,\bR)  
	\end{align*}
	since $\bProtMS^{(i)}$ is a function of $(\bX_i,\bsigma_i)$ where $\bX_i:= (\bX_{i,1},\ldots,\bX_{i,t})$ is a random variable for the vector $x^{(i)}$. 
	Moreover, $\bX_i$ defines the graph $G_i$ without the labels, i.e., over the set of vertices $V_i := [N]$ and $\bsigma_i$ is the random variable denoting how the vertices of the player $\PS{i}$ map to
        $G$, i.e., specify the labels of vertices. Therefore $(\bX_i,\bsigma_i)$ is independent of $\bJ = j$ (given the input graph $G_i$, each matching has the same probability
        of being the chosen matching for $\jstar$); hence it is easy to see that all four random variables in above term are independent of
        the event $\bJ = j$.  Moreover, since $\bY_j$ and $\bY^{-j}$ are independent of each other, conditioned on $\bsigma$ and $\bR$, 
        by conditional super-additivity of mutual information (\itfacts{info-super-additivity}), 
	\begin{align*}
	I(\bT; \bProtMSa \mid \bR,\bSigma,\bJ) &\leq\frac{1}{t} \sum_{i=1}^{k}  I(\bY_1,\ldots,\bY_t ; \bProtMS^{(i)} \mid \bSigma,\bR) \\
	&= \frac{1}{t} \sum_{i=1}^{k}  I(\bX_1,\ldots,\bX_k ; \bProtMS^{(i)} \mid \bSigma,\bR) \tag{$\bY_1,\ldots,\bY_t$ uniquely determines $\bX_1,\ldots,\bX_k$ and vice versa}\\
	&\leq  \frac{1}{t} \sum_{i=1}^{k} I(\bX_1,\ldots,\bX_k,\bsigma_1,\ldots,\bsigma_k ; \bProtMS^{(i)}, \bR) \tag{by chain rule of mutual information (\itfacts{info-chain-rule})} \\
	&= \frac{1}{t} \cdot \ICost{\protMSa}{\distMSa}
	\end{align*}
	where the last equality is because $(\bX_i,\bsigma_i)$ uniquely determines the input to the player \PS{i} for $i \in [k]$ and vice versa. Since 
	$I(\bT; \bProtMSa \mid \bR,\bSigma,\bJ) = \Omega(1)$, we obtain that $\ICost{\protMSa}{\distMSa} = \Omega(t)$.
\end{proof}

Theorem~\ref{thm:MSa-hard} now follows from Lemma~\ref{lem:alpha-info-cost} by noticing that $n = (2\alpha+1) \cdot N$, $r = 2\alpha N/k$ and $t \geq \frac{N^2}{2r} = \frac{N \cdot k}{2\alpha} = \Omega(nk/\alpha^2)$.


\section{Space Lower Bounds for $(1+\eps)$-Approximating Matching Size}\label{sec:eps-lower-bound}
In this section, we present our space lower bounds for algorithms that compute a $(1 + \eps)$-approximation of the maximum matching
size in graph streams. We first introduce some notation which will be used throughout this section.

\paragraph{Notation.} Fix any $(r,t)$-RS graph $\Grs(V,E)$ (for any parameters $r, t$) with induced matchings $\Mrs_1, \ldots, \Mrs_t$. For each
matching $\Mrs_i$, we assume an arbitrary ordering of the edges in $\Mrs_i$, denoted by $e_{i,1},\ldots,e_{i,r}$, and further denote
$e_{i,j} := (u_{i,j}, v_{i,j})$ for all $j \in [r]$.  Let $L(\Mrs_i) := \set{u_{i,1},\ldots,u_{i,r}}$ and
$R(\Mrs_i) := \set{v_{i,1},\ldots,v_{i,r}}$. We emphasize that we do not require $\Grs(V,E)$ to be necessarily a \emph{bipartite} graph; each bipartition 
$L(\Mrs_i)$ and $R(\Mrs_i)$ (for $i \in [t]$) is defined locally for the matching itself and hence a vertex $v$ is allowed to belong to, say, $L(\Mrs_i)$ and $R(\Mrs_j)$ for $i \neq j$, 
simultaneously. 

Furthermore, for each matching $\Mrs_i$ and any boolean vector $x \in \set{0,1}^r$, we define the matching $\Mrs_i|_x$ as the subset of (the
edges) of $\Mrs_i$ obtained by retaining the edge $e_{i,j} \in \Mrs_i$ (for any $j\in[r]$) iff $x(j) = 1$.  In addition, for the vertex set
$R(\Mrs_i)$ and any perfect $p$-hypermatching\footnote{Throughout this section, we use $p$ instead of the usual parameter $t$ for
  hypermatchings in order to avoid confusion with the parameter $t$ in RS graphs} $\HM$ on $[r]$, we define the \emph{$p$-clique family} of
$\HM$ on $R(\Mrs_i)$ to be a set of $\card{\HM}$ cliques where the vertices $C_{\mathbf{e}}$ of each clique is defined by a distinct
hyperedge $\mathbf{e} \in \HM$: $C_{\mathbf{e}} := \set{ v_{i,k} \mid k \in \mathbf{e}}$.


\subsection{Insertion-Only Streams}\label{sec:insert-matching-lower}

We define \MO{n,\eps} as the \emph{two-player one-way communication} problem of estimating the matching size to within a factor of
$(1+\eps)$, when Alice and Bob are each given a subset of the edges of an $n$-vertex input graph $G(V,E)$. In this section, we prove the
following lower bound on the information complexity of $\MO{n,\eps}$.

\begin{theorem}[Lower bound of \MO{n,\eps}]\label{thm:MO-hard}
  For any sufficiently large $n$ and sufficiently small $\eps < \frac{1}{2}$, there exists a distribution $\distMO$ for $\MO{n,\eps}$ such
  that for any constant $\delta < \frac{1}{2}$:
  \[
      \ICO{\MO{n,\eps}}{\distMO}{\delta} = \RS{n} \cdot n^{1-O(\eps)} 
  \]
\end{theorem}

The lower bound of theorem~\ref{thm:MO-hard}, together with Proposition~\ref{prop:cc-ic}, implies the same lower bound on the one-way
communication complexity of $\MO{n,\eps}$.  Since one-way communication complexity is a lower bound on the space complexity of any
single-pass streaming algorithm in insertion-only streams, this immediately proves Part~(1) of Theorem~\ref{thm:lower-eps-intro}.

In the following, we focus on proving Theorem~\ref{thm:MO-hard}. Suppose the maximum value for $\RS{n}$ is achieved by an $(r,t)$-RS graph
with the parameter $r = \cmax \cdot n$.  We propose the following (hard) input distribution $\distMO$ for $\MO{n,\eps}$.


\textbox{The hard distribution $\distMO$ for $\MO{n,\eps}$:}{
  \begin{itemize}
\item[]  \textbf{Parameters:} $N := \frac{n}{2-2\cmax}$, $r := \cmax \cdot N$, $t := \RS{N}$, and $p:= \floor{\frac{\cmax}{2\eps}}$.  
\item The input to the players is a graph $G(V,E_A \cup E_B)$ where $E_A$ is given to Alice and $E_B$ is given to Bob. 
\item \textbf{Alice:}
  \begin{enumerate}
  \item Let $V_1$ ($\subset V$) and $V_2:= V \setminus V_1$ be, respectively, a set of $N$ and $n-N$ vertices.
  \item Let $H$ be any fixed $(r,t)$-RS graph with $V(H) = V_1$.
  \item Draw $r$-dimensional binary vectors $x^{(1)}, \ldots, x^{(t)}$ \emph{independently} following the distribution $\distbhh$ for
    $\BHHZ{r}{p}$.
  \item The input to Alice is the edge-set $E_A := M_1 \cup \ldots \cup M_t$, where $M_j := \Mrs_j |_{x^{(j)}}$.
  \end{enumerate}
\item \textbf{Bob:}
  \begin{enumerate}
  \item Pick $\jstar \in [t]$ uniformly at random.
  \item For the vector $x^{(\jstar)}$, draw a perfect $p$-hypermatching $\HM$ following the distribution $\distbhh$ conditioned on
    $x^{(\jstar)}$; consequently, $(x^{(\jstar)}, \HM)$ is a $\BHHZ{r}{p}$ instance drawn from the distribution $\distbhh$.
  \item Let $E_{B,1}$ be an arbitrary perfect matching between $V_{1} \setminus V(\Mrs_{\jstar})$ and $V_2$.
  \item Let $E_{B,2}$ be the edges of the $p$-clique family of $\HM$ on $R(\Mrs_{\jstar})$.
  \item The input to Bob is the edge-set $E_B:= E_{B,1} \cup E_{B,2}$.
  \end{enumerate}
\end{itemize}
}

We say that the instance $(x^{(\jstar)}, \HM)$ of $\BHHZ{r}{p}$ in the distribution (denoted by $\Ibhh$) is \emph{embedded} inside $\distmatch$.
The following claim established the connection between $\Ibhh$ and maximum matching size in $G$. 

\begin{claim}\label{clm:insert-matching-bhh}
  Let: 
  \begin{align*}
  	\opt_{\Yes} &:= \min_{G}\Paren{\opt(G) \mid \text{$G$ is chosen from $\distMO$ \emph{conditioned} on $\Ibhh$ being a $\Yes$ instance}} \\
	\opt_{\No} &:= \max_{G}\Paren{\opt(G) \mid \text{$G$ is chosen from $\distMO$ \emph{conditioned} on $\Ibhh$ being a $\No$ instance}}
  \end{align*}
  then, $\paren{1-\eps} \cdot \opt_{\Yes} > \opt_{\No}$.  
\end{claim}
\begin{proof}
  Let $\Mstar$ be a maximum matching in $G$. Since all vertices in $V_2$ have degree $1$, without loss of generality, we can assume $\Mstar$
  contains the matching $E_{B,1}$ between $V_{1} \setminus V(\Mrs_{\jstar})$ and $V_2$. Consequently the size of $\Mstar$ only depends on
  how many vertices in $V(\Mrs_{\jstar})$ can be matched with each other. 
  
  Consider the subgraph $H:=G[L(\Mrs_{\jstar}) \cup R(\Mrs_{\jstar})]$ of $G$; by Claim~\ref{clm:bhh-matching}, if $\Ibhh$ is a 
  \Yes instance, then $\opt(H) = \frac{3r}{4}$. Hence, in this case, 
  \[
  	\opt(G) = \card{V_2} + \opt(H) = N-2r + \frac{3r}{4} = N - \frac{5\cmax N}{4} = \frac{4-5\cmax}{4}\cdot N 
  \]
  If $\Ibhh$ is a \No instance, then $\opt(H) = \frac{3r}{4} - \frac{r}{2p}$. Hence, in this case, 
  \[
	\opt(G) = \card{V_2} + \opt(H) \leq N-2r + \frac{3r}{4} - \frac{r}{2p}= N - \frac{5\cmax N}{4} - \frac{\cmax}{2p} N = \frac{4-5\cmax}{4}\cdot N -\frac{\cmax}{2p} N
  \]
  The bound on $\opt_{\Yes}$ and $\opt_{\No}$ now follows from the fact that $p \leq \frac{\cmax}{2\eps}$ and therefore
  $\frac{\cmax}{2p} N \ge \eps N > \eps \cdot \paren{ \frac{4-5\cmax}{4}\cdot N}$.
\end{proof}

Fix any $\delta$-error protocol $\protMO$ for $\MO{n,\eps}$ on $\distMO$; Claim~\ref{clm:insert-matching-bhh}
implies that $\protMO$ is also a $\delta$-error protocol for solving the embedded instance $\Ibhh$: simply return \Yes
whenever the estimate is larger than $\opt_{\No}$ and return \No otherwise.  We now use this fact to design a protocol $\protBHH$ for solving
$\BHHZ{r}{p}$ on $\distbhh$, and prove that the information cost of $\protMO$ is $t$ times the information cost of $\BHHZ{r}{p}$.

\paragraph{The protocol $\protBHH$ for reducing $\BHHZ{r}{p}$ to $\Matching_{n,\eps}$:}
\begin{enumerate}
\item Let $(x,\HM)$ be the input $\BHHZ{r}{p}$ instance ($x$ is given to Alice and $\HM$ is given to Bob).
	\item Using \emph{public randomness}, Alice and Bob sample an index $\jstar \in [t]$ uniformly at random.
	\item Let $x^{(1)},\ldots,x^{(t)}$ be $t$ vectors in $\set{0,1}^r$ whereby $x^{(\jstar)} = x$ and for any $j \neq \jstar$, $x^{(j)}$ is sampled by Alice 
	using \emph{private randomness} as in the distribution $\distMO$. Alice creates the edges $E_A$ following the distribution $\distMO$ using these vectors. 
	
	\item Given the $p$-hypermatching $\HM$ as input, Bob creates $E_{B,1}$ as an arbitrary perfect matching between $V_{1} \setminus V(\Mrs_{\jstar})$ and $V_2$. He also creates
	$E_{B,2}$ as the edges of the $p$-clique family of $\HM$ on $R(\Mrs_{\jstar})$ ($V_1$, $V_2$, and $\Mrs_{\jstar}$ are defined exactly as in $\distMO$). 
	\item The players then run $\protMO$ on the graph $G(V,E_A \cup E_B)$ and Bob outputs $\Yes$ if the output is larger than $\opt_{\No}$ and $\No$ otherwise.
\end{enumerate}

The correctness of the protocol follows immediately from Claim~\ref{clm:insert-matching-bhh}. We now bound the information cost of this new protocol. 

\begin{lemma}\label{lem:insert-direct-sum}
	$\ICost{\protBHH}{\distbhh} \leq \frac{1}{t} \cdot  \ICost{\protMO}{\distMO}$. 
\end{lemma}  
\begin{proof}
	We have, 
	\begin{align*}
		\ICost{\protBHH}{\distbhh} &= I_{\distbhh}(\bX ; \bProtBHH,\bR) = I_{\distbhh}(\bX ; \bProtBHH^{R} \mid \bR) \tag{by chain rule of mutual information (\itfacts{info-chain-rule}) and since $I(\bX;\bR) = 0$ as $\bX \perp \bR$} \\
		&= I_{\distbhh}(\bX ; \bProtMO \mid \bJ) \tag{$\bR = \bJ$ and the message of $\protBHH$ is the same as $\protMO$ after fixing the index $\jstar$} \\
		&= \Ex_{j \in [t]} \bracket{ I_{\distbhh}(\bX ; \bProtMO \mid \bJ = j)} = \frac{1}{t} \cdot \sum_{j=1}^{t} I_{\distbhh}(\bX_j ; \bProtMO \mid \bJ = j) \\
		&= \frac{1}{t} \cdot \sum_{j=1}^{t} I_{\distMO}(\bX_j ; \bProtMO \mid \bJ = j) \tag{joint distribution of $\bProtMO$ and $\bX_j$, conditioned on $\bJ=j$, is the same under $\distMO$ and $\distbhh$} \\
		&= \frac{1}{t} \cdot \sum_{j=1}^{t} I_{\distMO}(\bX_j ; \bProtMO)
	\end{align*}
	where the last equality is true since the random variables $\bX_j$ and $\bProtMO$ are both independent of the event $\bJ = i$ (by definition of the distribution $\distMO$).
	Finally, 
	 \begin{align*}
	 	\ICost{\protBHH}{\distbhh} &= \frac{1}{t} \cdot \sum_{j=1}^{t} I_{\distMO}(\bX_j ; \bProtMO) \\
		&\leq \frac{1}{t} \cdot I_{\distMO}(\bX_1,\ldots,\bX_t; \bProtMO) \tag{by conditional super-additivity of mutual information (\itfacts{info-super-additivity}) since $\bX_j \perp \bX^{<j}$} \\
		&=  \frac{1}{t} \cdot I_{\distMO}(\bE_A; \bProtMO) = \frac{1}{t} \cdot  \ICost{\protMO}{\distMO}
	 \end{align*}
	 where the second last inequality is because the set of edges in $E_A$ can be determined uniquely by the vectors $x^{(1)},\ldots,x^{(t)}$ and vice versa. 
\end{proof}

Theorem~\ref{thm:MO-hard} now follows from Lemma~\ref{lem:insert-direct-sum}, lower bound of $\Omega(r^{1-1/p}) = n^{1-O(\eps)}$ for 
$\BHHZ{r}{p}$ in Corollary~\ref{cor:bhh}, and the choice of $t = \RS{n}$.


\subsection{Dynamic Streams}\label{sec:dynamic-matching-lower}

We define \MS{n,k,\eps} as the \emph{$k$-player simultaneous communication} problem of estimating the maximum matching size to within a
factor of $(1+\eps)$, when edges of an $n$-vertex input graph $G(V,E)$ are partitioned across the $k$-players and the referee (see
Remark~\ref{rem:referee}).  In this section, we prove the following lower bound on the information complexity of $\MS{n,k,\eps}$ in the SMP
communication model.

\begin{theorem}[Lower bound for \MS{n,k,\eps}]\label{thm:MS-hard}
  For any sufficiently large $n$ and sufficiently small $\eps < \frac{1}{2}$, there exists some $k = n^{o(1)}$ and a distribution
  $\distMS$ for $\MS{n,k,\eps}$ such that for any constant $\delta < \frac{1}{2}$:
  \[
      \ICS{\MS{n,k,\eps}}{\distMS}{\delta} = n^{2-O(\eps)} 
  \]
\end{theorem}

Theorem~\ref{thm:MS-hard}, combined with Proposition~\ref{prop:cc-ic}, immediately proves the same lower bound on the SMP communication
complexity of $\MS{n,k,\eps}$. Since SMP communication complexity of a $k$-player problem is at most $k$ times the space complexity of any
single-pass streaming algorithm in dynamic streams~\cite{LiNW14,AiHLW16} (and $k = n^{o(1)}$), this immediately proves Part~(2) of
Theorem~\ref{thm:lower-eps-intro}.


In the following, we focus on proving Theorem~\ref{thm:MS-hard}. We propose the following (hard) distribution $\distMS$ for $\MS{n,k,\eps}$.
Intuitively, the distribution $\distMS$ can be seen as imposing the hard distribution for matching size estimation in~\cite{BuryS15} on each
induced matching in the hard instance of~\cite{AssadiKLY16} for finding approximate matchings.

\textbox{The hard distribution $\distMS$ for $\MS{n,k,\eps}$:}{
\begin{enumerate}
\item[] \textbf{Parameters:} $r = N^{1-o(1)},~t = \frac{{N \choose 2} - o(N^2)}{r},~k = \frac{N}{\eps \cdot r},~n=N + k\cdot r$, and
  $p:= \floor{\frac{1}{8\eps}}$.
\item Fix an $(r,t)$-RS graph $\Grs$ on $N$ vertices.
\item Pick $\jstar \in [t]$ uniformly at random and draw a $\BHHZ{r}{p}$ instance $(x^{(\jstar)}, \HM)$ from the
  distribution $\distbhh$.
\item For each player \PS{i}  independently,
  \begin{enumerate}
  \item Denote by $G_i$ the input graph of \PS{i}, initialized to be a copy of $\Grs$ with vertices $V_i = [N]$.
  \item Let $\vstari$ be the set of vertices matched in the $\jstar$-th induced matching of $G_i$. Change the induced matching
    $\Mrs_{\jstar}$ of $G_i$ to $M_{\jstar} := \Mrs_{\jstar} |_{x^{(\jstar)}}$.
  \item For any $j \in [t]\setminus \set{\jstar}$, draw a vector $x^{(i,j)} \in \set{0,1}^{r}$ following the
    distribution $\distbhh$ for $\BHHZ{r}{p}$, and change the induced matching $\Mrs_j$ of $G_i$ to $M_j := \Mrs_j |_{x^{(j)}}$.
  \item Create the $p$-clique family of $\HM$ on the vertices $R(\Mrs_{\jstar})$, and give the edges of the $p$-clique family to the referee.
  \end{enumerate}
\item Pick a random permutation $\sigma$ of $[n]$. For every player $\PS{i}$, for each vertex $v$ in $V_i\setminus \vstari$ with label $j$
  ($\in [N]$), \emph{relabel} $v$ to $\sigma(j)$. Enumerate the vertices in $\vstari$ (from the one with the smallest label to the largest),
  and relabel the $j$-th vertex to $\sigma(N + (i-1) \cdot 2r + j)$.  In the final graph, the vertices with the same label correspond to
  the same vertex.
\end{enumerate}
} 


The vertices whose labels belong to $\sigma([N])$ are referred to as \emph{shared} vertices since they belong to the input graph of \emph{every}
player, and the vertices $\vstari$ are referred to as the \emph{private} vertices of the player $\PS{i}$ since they only appear in
the input graph of $\PS{i}$ (in the final graph, i.e., after relabeling). We point out that, in general, the final graph constructed by this distribution is a 
multi-graph with $n$ vertices and $O(kN^2) = O(n^{2})$ edges (counting the multiplicities); the multiplicity of each edge is also at most $k$. 
Finally, the existence of an $(r,t)$-RS graph $\Grs$ with the parameters used in this distribution is guaranteed by a result of~\cite{AlonMS12} (see Section~\ref{sec:rs-graphs}).

Similar to the lower bound in Section~\ref{sec:insert-matching-lower}, let $\Ibhh$ be the \emph{embedded} $\BHHZ{r}{p}$ instance
$(x^{(i)}, \HM)$.  The following claim is analogous to Claim~\ref{clm:insert-matching-bhh} in Section~\ref{sec:insert-matching-lower}.

\begin{claim}\label{clm:dynamic-matching-bhh}
  Let: 
  \begin{align*}
  	\opt_{\Yes} &:= \min_{G}\Paren{\opt(G) \mid \text{$G$ is chosen from $\distMS$ \emph{conditioned} on $\Ibhh$ being a $\Yes$ instance}} \\
	\opt_{\No} &:= \max_{G}\Paren{\opt(G) \mid \text{$G$ is chosen from $\distMS$ \emph{conditioned} on $\Ibhh$ being a $\No$ instance}}
  \end{align*}
  then, $\paren{1-\eps} \cdot \opt_{\Yes} > \opt_{\No}$.  
\end{claim}

\begin{proof}
  We partition the edges of $G$ into $k+1$ groups: for any $i \in [k]$, group $i$ contains the edges that are between the private vertices
  $\vstari$ of player $\PS{i}$, and group $k+1$ contains the edges incident on at least one shared vertex. Let $H_i := G[\vstari]$, i.e.,
  the subgraph of $G$ induced on the vertices $\vstari$.

  
  If $\Ibhh$ is a \Yes instance, then for any $i \in [k]$, $\opt(H_i) = \frac{3r}{4}$ by Claim~\ref{clm:bhh-matching}.  Since $\vstari$ are
  private vertices, one can choose \emph{any} matching from each $H_i$, and the collection of the chosen edges form a matching of
  $G$. Therefore,
  \begin{align*}
  	\opt(G) > \sum_{i=1}^{k} \opt(H_i) = \frac{3kr}{4} = \frac{3N}{\eps}
  \end{align*} 
  Note that, $\opt(G)$ is \emph{strictly} larger than $\frac{3N}{\eps}$ since one can add (any) edge between the public vertices to the
  matching.

  If $\Ibhh$ is a \No instance, then $\opt(H_i) = \frac{3r}{4} - \frac{r}{2p}$. Since the maximum matching size in $G$ is at most the
  summation of the maximum matching size in each group, we have
  \begin{align*}
  	\opt(G) \leq \sum_{i=1}^{k} \opt(H_i) + N \leq \frac{3kr}{4} - \frac{kr}{2p} + N \le  \frac{3N}{\eps} - {3N}
  \end{align*} 
  and the gap between $\opt_{\Yes}$ and $\opt_{\No}$ follows.
\end{proof}

Fix any $\delta$-error protocol $\protMS$ for $\MS{n,k,\eps}$ on $\distMS$; Claim~\ref{clm:dynamic-matching-bhh} implies that
$\protMS$ is also a $\delta$-error protocol for solving the embedded instance $\Ibhh$: simply return \Yes whenever the estimate is larger
than $\opt_{\No}$ and return \No otherwise.  In the following, we use this fact to design a protocol $\protBHH$ for solving $\BHHZ{r}{p}$ on
$\distbhh$, and then prove that the information cost of $\protMS$ is $t$ times the information cost of $\BHHZ{r}{p}$.

In the protocol $\protBHH$, Alice will simulate all $k$ players of $\MS{n,k,\eps}$ and Bob will simulate the referee; Alice and
Bob will use public coins to draw the special index $\jstar$ and the permutation $\sigma$. Together with the input from $\distbhh$,
Alice and Bob will be able to create a $\MS{n,k,\eps}$ instance. The reduction is formally defined as follows (the parameters used in the
reduction are exactly the same as that in the definition of $\distMS$).


\paragraph{The protocol $\protBHH$ for reducing $\BHHZ{r}{p}$ to $\Matching_{n,k,\eps}$:}
\begin{enumerate}
	\item Let $(x,\HM)$ be the input $\BHHZ{r}{p}$ instance ($x$ is given to Alice and $\HM$ is given to Bob).
	\item Using \emph{public randomness}, Alice and Bob sample an index $\jstar \in [t]$, and a permutation $\sigma$ on $[n]$ uniformly at random.
	\item For any player $\PS{i}$, let $x^{(i,1)},\ldots,x^{(i,t)}$ be $t$ vectors in $\set{0,1}^r$ whereby $x^{(i,\jstar)} = x$ (i.e., Alice's input in the \BHHZ{r}{p} problem) 
	and for any $j \neq \jstar$, $x^{(i,j)}$ is sampled by Alice using \emph{private randomness} as in the distribution $\distMS$. Alice then uses these vector together with permutation $\sigma$ to 
	create the input graph $G_i$ for each player \PS{i} for $i \in [k]$ following how $G_i$ is created in the distribution $\distMS$ for $\MS{n,k,\eps}$. 
	
	\item The vertices $R(\Mrs_{\jstar})$ of each player will be mapped (by $\sigma$) to a different set of vertices in $G$.  Since Bob
          knows $\sigma$ and $\jstar$, and the (input) $p$-hypermatching $\HM$, Bob can create the $p$-clique families of each player (following the
          input of the referee in $\distMS$).
	\item The players then run $\protMS$ on the $\MS{n,k,\eps}$ that they created, and Bob outputs $\Yes$ if the matching size estimate
          is larger than $\opt_{\No}$ and $\No$ otherwise.
\end{enumerate}

It is straightforward to verify that the distribution of the $\MS{n,k,\eps}$ instance created by the protocol $\protBHH$ is identical to the
distribution $\distMS$. The correctness of the protocol now follows immediately from Claim~\ref{clm:insert-matching-bhh}. In the remainder of this section, we bound the
information cost of this protocol.

\begin{lemma}\label{lem:dynamic-direct-sum}
	$\ICost{\protBHH}{\distbhh} \leq \frac{1}{t} \cdot  \ICost{\protMS}{\distMS}$. 
\end{lemma}  
\begin{proof}
	We have, 
	\begin{align*}
		\ICost{\protBHH}{\distbhh} &= I_{\distbhh}(\bX; \bProtBHH,\bR) = I_{\distbhh}(\bX ; \bProtBHH^{R} \mid \bR) 
		 \tag{by chain rule of mutual information (\itfacts{info-chain-rule}) and since $I(\bX;\bR) = 0$ as $\bX \perp \bR$} \\
		&= I_{\distbhh}(\bX ; \bProtMS \mid \bSigma,\bJ,\bRms) = \Ex_{j \in [t]} \bracket{I_{\distbhh}(\bX ; \bProtMS \mid \bSigma,\bRms,\bJ=j)} 
	\end{align*}
	where the second last equality is because $\bR = (\bsigma,\bJ,\bRms)$ ($\bRms$ is the public randomness of $\protMS$), and the message of $\protBHH$ is the same as $\protMS$ after fixing the index $\jstar$. 
	For any $j \in [t]$, define $\bY_j := (\bX_{1,j},\bX_{2,j},\ldots,\bX_{k,j})$ where $\bX_{i,j}$ is a random variable for the vector $x^{(i,j)}$. With this notation, conditioned on $\bJ=j$, 
	we have $\bX = \bY_j$ and also the joint distribution of $(\bProtMS, \bSigma, \bY_j,\bRms)$ conditioned on $\bJ=j$, is the same under both $\distMS$ and $\distbhh$. Hence, 
	\begin{align*}
	\ICost{\protBHH}{\distbhh} &= \frac{1}{t} \sum_{j=1}^{t} I_{\distMS}(\bY_j ; \bProtMS \mid \bSigma,\bRms,\bJ=j) \\
		&= \frac{1}{t} \sum_{j=1}^{t} I_{\distMS}(\bY_j ; \bProtMS^{(1)},\ldots,\bProtMS^{(k)} \mid \bSigma,\bRms,\bJ=j) \\
		&\leq \frac{1}{t} \sum_{j=1}^{t} \sum_{i=1}^{k}  I_{\distMS}(\bY_j;\bProtMS^{(i)}\mid \bSigma,\bRms,\bJ=j) \\
		&= \frac{1}{t} \sum_{j=1}^{t} \sum_{i=1}^{k} I_{\distMS}(\bY_j;\bProtMS^{(i)}\mid \bSigma,\bRms) 
		\end{align*}
	where the last inequality is by conditional sub-additivity of mutual information (\itfacts{info-sub-additivity}) since $\bProtMS^{(i)} \perp \bProtMS^{<i} \mid \bSigma,\bY_j,\bRms,\bJ=j$; this is 
	because conditioned on the given random variables and $\bJ = j$, the message of each player $\PS{i}$ (i.e., $\bProtMS^{(i)}$) is only a function of $x^{(i,j)}$ for $j \neq \jstar$ and since these
	vectors are chosen independently, the messages would be independent. 
	
	Moreover, the reason we can drop the conditioning on the event $\bJ = j$ (in the last equality above) is as follows: $\bProtMS^{(i)}$ is a function of $(\bX_i,\bsigma_i)$  
	where $\bX_i:= (\bX_{i,1},\ldots,\bX_{i,t})$ is a random variable for the vector $x^{(i)}$. $\bX_i$ defines the graph $G_i$ without the labels, i.e., 
	over the set of vertices $V_i := [N]$ and $\bsigma_i$ is the random variable denoting how the vertices of the player $\PS{i}$ are mapped to
        $G$, i.e., specifies the labels of vertices. Therefore, $(\bX_i,\bsigma_i)$ is independent of $\bJ = j$ (given the input graph $G_i$, each matching has the same probability
        of being the chosen matching for $\jstar$); hence it is easy to see that all four random variables in above term are independent of
        the event $\bJ = j$. 
        
        Finally, since $\bY_j$ and $\bY^{-j}$ are independent of each other even conditioned on $\bSigma,\bRms$, by conditional super-additivity of 
        mutual information (\itfacts{info-super-additivity}), 
	\begin{align*}
		\ICost{\protBHH}{\distbhh}  &\leq \frac{1}{t} \sum_{i=1}^{k} I_{\distMS}(\bY_1,\ldots,\bY_t;\bProtMS^{(i)} \mid \bSigma,\bRms,) \\
		&= \frac{1}{t} \sum_{i=1}^{k} I_{\distMS}(\bX_1,\ldots,\bX_k;\bProtMS^{(i)} \mid \bSigma,\bRms,) \tag{$\bY_1,\ldots,\bY_t$ uniquely define $\bX_1,\ldots,\bX_k$ and vice versa}\\
		&\leq \frac{1}{t} \sum_{i=1}^{k} I_{\distMS}(\bX_1,\ldots,\bX_k,\bsigma;\bProtMS^{(i)},\bRms) \tag{by chain rule of mutual information (\itfacts{info-chain-rule})} \\
		&= \frac{1}{t} \cdot \ICost{\protMS}{\distMS}
	\end{align*}
	where the last equality is because $(\bX_i,\bsigma_i)$ uniquely defines the input to player $\PS{i}$. 
\end{proof}

Theorem~\ref{thm:MS-hard} now follows from Lemma~\ref{lem:dynamic-direct-sum}, lower bound of $\Omega(r^{1-1/p}) = n^{1-O(\eps)}$ for 
$\BHHZ{r}{p}$ in Corollary~\ref{cor:bhh}, and the choice of $t = \Theta(n)$ in the distribution.

\section{Space Upper Bounds for  $\alpha$-Approximating Matching Size}\label{sec:alpha-upper}
In this section, we present our algorithms for achieving an $\alpha$-approximation of the maximum matching size respectively in
$\Ot({n / \alpha^2})$ space for insertion-only streams and in $\Ot({n^2 / \alpha^4})$ space for dynamic streams, proving Theorem~\ref{thm:upper-intro}.



The main ingredient of both our algorithms is a simple \emph{vertex sampling} procedure.  In the rest of this section, we first define this
sampling procedure and establish its connection to matching size estimation (Section~\ref{sec:vertex-sample}).  We then build on this
connection to provide a \emph{meta-algorithm} for matching size estimation (Section~\ref{sec:meta-alg}). Finally, we show how to implement
this meta-algorithm in $\Ot(n/\alpha^2)$ space in insertion-only streams and $\Ot(n^2/\alpha^4)$ space in dynamic streams, which proves
Theorem~\ref{thm:upper-intro}.


\subsection{Vertex Sampling Procedure}\label{sec:vertex-sample}

Consider the following simple vertex sampling procedure. 
\begin{itemize}
	\item[] \sample{G}{p}: sample each vertex $v \in V$ in $G(V,E)$ w.p. $p$, using a \emph{four-wise independent} hash function, and return the induced subgraph over the set of sampled vertices. 
\end{itemize}

Note that since $O(\log{n})$ bits suffices to store a four-wise independent hash function (see, e.g.,~\cite{RAbook}), the set of sampled vertices in \sample{G}{p} can be also 
be stored (implicitly) in $O(\log{n})$ bits. 

The following lemma establishes that as long as $\opt(G)$ is not too small, the maximum matching size in the graph that \sample{G}{p}
outputs (for $p := \frac{\log{n}}{\alpha}$) can be directly used to obtain an $\alpha$-approximation of $\opt(G)$.


\begin{lemma}\label{lem:vertex-sample}
  Let $G(V,E)$ be any graph, $\alpha \geq \log{n}$, and $p:= \frac{\log{n}}{\alpha}$; for $\GS := \sample{G}{p}$,
	\begin{enumerate}
		\item \label{part:sample-upper} if $\opt(G) = \Omega(\alpha)$, then $\opt(\GS) \le {3 \log n \over  \alpha} \cdot \opt(G)$ w.p. $1-o(1)$. 
		\item \label{part:sample-lower} if $\opt(G) = \Omega(\alpha^2)$, then $\opt(\GS) \ge {\log^2 n \over 2\alpha^2} \cdot
                  \opt(G)$ w.p. $1-o(\frac{1}{\log{n}})$.
	\end{enumerate}
\end{lemma}

Note that for $\opt(G) = \Omega(\alpha^2)$, Lemma~\ref{lem:vertex-sample} immediately implies that w.p. $1 - o(1)$,
\begin{align*}
  {3 \log n \over \alpha} \cdot \opt(G) \le \opt(\GS) \le {\log^2 n \over 2\alpha^2} \cdot \opt(G).
\end{align*}

\begin{proof}[Proof of Lemma~\ref{lem:vertex-sample}]
  Fix a maximum matching $\Mstar$ in $G$ and denote the set of vertices matched in $\Mstar$ by $V(\Mstar)$. 
  Moreover, let $\VS(\Mstar)$ be the set of vertices in $V(\Mstar)$ that are sampled by \sample{G}{p}. 
  
  We first prove Part~(\ref{part:sample-upper}) of the lemma. Since $\Mstar$ is a maximum matching in $G$, every edge in $G$ must be
  incident on at least one vertex in $V(\Mstar)$.  Consequently, in the sampled graph $\GS$, every edge is incident on at least one vertex
  in $\VS(\Mstar)$, and hence, $\opt(\GS) \le \card{\VS(\Mstar)}$; therefore, we only need to upper bound $\card{\VS(\Mstar)}$. Let $X$ be a
  random variable denoting $\card{\VS(\Mstar)}$.
  
  Now, $\expect{X} = p \cdot \card{V(\Mstar)} = p \cdot 2\opt(G) = {2\log n \over \alpha} \cdot \opt(G)$ by the choice of $p$.  Since
  $\opt(G) = \Omega(\alpha)$ by our assumption in Part~(\ref{part:sample-upper}), $\expect{X} = \Omega(\log n)$. Moreover, because
  \sample{G}{p} samples vertices using a four-wise independent hash function, $\var{X} \leq \expect{X}$ and hence by Chebyshev inequality,
  \begin{align*}
    \prob{X \geq {3 \log n \over \alpha} \cdot \opt(G)} &=  \prob{X \geq \frac{3}{2} \cdot \expect{X}} \leq  \prob{\card{X - \expect{X}} \ge {\expect{X} \over 2}} \\
    &\leq {\var{X} \over (\expect{X}/2)^2} \le {4 \over \expect{X}} = {1 \over
    \Omega(\log n)} = o(1)
  \end{align*}
  This implies w.p. $1 - o(1)$, $\card{\VS(\Mstar)} \le {3 \log n \over \alpha} \cdot \opt(G)$, and 
  since $\opt(\GS) \leq \card{\VS(\Mstar)}$ we obtain the result in Part~(\ref{part:sample-upper}). 
  
  We now prove Part~(\ref{part:sample-lower}) of the lemma. Let $\MstarS$ be the set of sampled edges from $\Mstar$ that end up $\GS$.
  Since $\opt(\GS) \geq \card{\MstarS}$, it suffices to show that $\card{\MstarS} \ge {\log^2 n \over 2\alpha^2} \cdot \opt(G)$.  Let $Y$ be
  a random variable denoting $\card{\MstarS}$.

  For each edge $e \in \Mstar$, $e$ appears in $\MstarS$ iff both endpoints of $e$ are sampled by \sample{G}{p}, which happens w.p. $p^2$
  (due to four-wise independence in sampling vertices). Therefore, the expected number of edges in $\MstarS$ is
  $\expect{Y}=p^2 \cdot \opt(G) = {\log^2 n \over \alpha^2} \cdot \opt(G)$. Since by assumption in Part~(\ref{part:sample-lower}),
  $\opt(G) = \Omega(\alpha^2)$, we have $\expect{Y} = \Omega(\log^2 n)$. Moreover, since vertices are sampled in \sample{G}{p} using a
  four-wise independent hash function, for any two edges in $\Mstar$, the event that they appear in $\MstarS$ is independent of each other;
  this implies $\var{Y} \le \expect{Y}$, and hence by Chebyshev inequality,
  \begin{align*}
    \prob{Y < {\log^2 n \over 2\alpha^2} \cdot \opt(G)} &= \prob{Y < {\expect{Y} \over 2}} \leq \prob{\card{Y - \expect{Y}} \ge {\expect{Y} \over 2}} \\
    &\leq {\var{Y} \over (\expect{Y}/2)^2} \le {4 \over \expect{Y}} = {1 \over  \Omega(\log^2 n)} = o({1 \over \log n})
  \end{align*}
  Therefore, w.p. $1-o({1 \over \log n})$, $\card{\MstarS} \ge {\log^2 n \over 2\alpha^2} \cdot \opt(G)$, proving Part~(\ref{part:sample-lower}). 
\end{proof}

\subsection{The Meta Algorithm}\label{sec:meta-alg}

In this section, we define our meta-algorithm for approximating the matching size in any graph $G$ based on the vertex sampling procedure
defined in the previous section.  To continue, we need to define the notion of \emph{matching size testers} that are used as subroutines in
the meta-algorithm.

\begin{definition}[$\gamma$-Matching Size Tester]
  For any constant $0 < \gamma < 1$, a $\gamma$-matching size tester (denoted by $\testerg$) is an algorithm that given a graph $G$ and a threshold $k$,
  outputs \Yes if $\opt(G) \ge k$, outputs \No if $\opt(G) \le \gamma \cdot k$, and otherwise is allowed to output either \Yes or \No. 
  
  Moreover, whenever $\testerg(G,k)$ outputs $\No$, it also outputs an estimate $\topt$ such that
  $\gamma \cdot \opt(G) \le \topt \le \opt(G)$.
\end{definition}

Given any $\gamma$-matching size tester $\testerg$, consider the following algorithm (denoted by Algorithm~1) for achieving an
$O(\alpha)$-approximation of maximum matching size.

  \begin{enumerate}
  \item\label{step:guesses} For each value $\beta \in \set{\log n, 2 \log n, 2^2 \log n, \ldots, \alpha}$, let $\galphap := \sample{G}{\log n\over \beta}$. 
  In parallel, run $\testerg$ on each $\galphap$ with the parameter ${\log^2 n \over 2}$ (i.e., run
    $\testerg(\galphap, {\log^2 n \over 2})$).
  \item\label{step:large-test} In addition, for $\beta = \alpha$, also run $\testerg(\galpha, {n \log^2 n \over \alpha^2})$.
  \item\label{step:output} At the end of the stream, for each value $\beta$, we say $\beta$ \emph{passes} if $\testerg(\galphap, {\log^2n \over 2})$ outputs
    \Yes; otherwise, we say $\beta$ \emph{fails}.
    \begin{itemize}
    \item If all $\beta$ fail, output the estimate $\tOPT_{\log n}$ returned by $\testerg(G^{\log n}, {\log^2 n \over 2})$.
    \item If all $\beta$ pass, output $\max\set{\alpha, {{\alpha \over \log^2n} \cdot \tOPT_\alpha}}$, where $\tOPT_{\alpha}$ is defined as
      follows. if $\testerg(\galpha, {n \log^2 n \over \alpha^2})$ returns \No, let $\tOPT_\alpha$ be the estimate returned by
      $\testerg(\galpha, {n \log^2 n \over \alpha^2})$; otherwise, let $\tOPT_\alpha := {\gamma n \log^2n \over \alpha^2}$.
    \item Otherwise, output ${\beta^* \over 2}$ where $\beta^*$ is the smallest $\beta$ that fails.
    \end{itemize}
  \end{enumerate}
  
  We should remark right away that if $\opt(G) = \Omega(\alpha^2)$, running $\testerg(\galpha, {n \log^2 n \over \alpha^2})$
  (step~\ref{step:large-test} in the algorithm) suffices to obtain an $\alpha$-approximation (Lemma~\ref{lem:vertex-sample} essentially guarantees that $\opt(\galpha) \in [{\opt(G) \over \alpha^2}, {\opt(G) \over \alpha}]$).
  Therefore, running tester for $O(\log{\alpha})$ different values (step~\ref{step:guesses} in the algorithm) is only for the case where
  $\opt(G) \le \alpha^2$.

  Intuitively speaking, for the three cases that determine the output of the algorithm (step~\ref{step:output} in the algorithm):
  \begin{itemize}
  \item If all $\beta$ fails, then all testers returns \No, which means the maximum matching size in the sampled graphs are all small: this
    is for the case $\opt(G) = \Ot(1)$.
  \item If all $\beta$ passes, then all testers return \Yes, which means the maximum matching sizes are all large: this is for the case
    $\opt(G) > \alpha^2$.
  \item Finally, if some $\beta$ pass and some $\beta$ fail, then we are in the case $\opt(G) \in [\Ot(1), \alpha^2]$.
  \end{itemize}
  We now prove the correctness of Algorithm~1 through considering these three cases separately.

  \begin{lemma}\label{lem:alpha-apx}
    For any $\alpha \ge \log n$, Algorithm~1 outputs an $O(\alpha)$-approximation of $\opt(G)$ w.h.p.
  \end{lemma}
  \begin{proof}
    First notice that if all $\beta$ fails, in particular, $\beta = \log n$ fails, and hence the estimate returned by
    $\testerg(G^{\log n}, \log^2 n)$ (denoted by $\tOPT_{\log n}$) is a $\gamma$-approximation of $\OPT(G^{\log n})$. Furthermore, note that
    for $\beta = \log n$, the subsampling probability is ${\log n \over \beta} = 1$, and hence, $G^{\log n} = G$. Therefore,
    $\tOPT_{\log n}$ is a also $\gamma$-approximation of $\opt(G)$.

    In the following, we analyze the other two cases: $(i)$ all $\beta$ pass (which would be the case where $\opt(G)$ is large) and $(ii)$
    some $\beta$ pass and some $\beta$ fail (which will be the case where $\opt(G)$ is small).  The following two claims summarize the
    property of the $\beta$ that passes and the property of the $\beta$ that fails, which will be useful for the analysis.

  \begin{claim}\label{clm:pass}
    For any $\beta$ where $\beta^2 \le \opt(G)$, $\beta$ passes w.p. $1 - o({1 \over \log n})$.
  \end{claim}
  \begin{proof}
    By Lemma~\ref{lem:vertex-sample} Part~(\ref{part:sample-lower}), when $\beta^2 \le \opt(G)$, w.p. $1 - o({1 \over \log n})$,
    \begin{align*}
      \opt(\galphap)  &\ge { \log^2 n \over 2\beta^2 } \cdot \opt(G) \ge { \log^2 n \over 2\beta^2 } \cdot \beta^2 = {\log^2 n \over 2}.
    \end{align*}
    Therefore, $\testerg(\galphap, {\log^2 n \over 2})$ outputs \Yes (and hence $\beta$ passes).
  \end{proof}

  \begin{claim}\label{clm:fail}
    For the value $\beta$ where ${\beta \over 2} \le \opt(G) \le \beta$ (if one exists), $\beta$ fails w.p. $1 - o(1)$.
  \end{claim}
  \begin{proof}
    By Lemma~\ref{lem:vertex-sample} Part~(\ref{part:sample-upper}), when $\opt(G) \ge {\beta \over 2}$, w.p. $1 - o(1)$, 
    \begin{align*}
      \opt(\galphap) & \le {3\log n \over \beta} \cdot  \opt(G) \le {3\log n \over \beta} \cdot \beta = 3\log n < {\gamma\log^2 n \over 2 } ~~
                        (\text{for $n$ sufficiently large}).
    \end{align*}
    Therefore, $\testerg(\galphap, {\log^2 n \over 2})$ outputs \No (and hence $\beta$ fails).
  \end{proof}

  With Claim~\ref{clm:pass} and Claim~\ref{clm:fail}, the correctness of case~$(ii)$ follows immediately.
  \begin{lemma}\label{lem:case-ii}
    If some $\beta$ pass and some $\beta$ fail, then ${\beta^* \over 2}$ is an $O(\alpha)$-approximation of $\opt(G)$.
  \end{lemma}
  \begin{proof}
    We first show that ${\beta^* \over 2} \ge {\opt(G) \over 2\alpha}$ and then show that ${\beta^* \over 2} \le \opt(G)$. To see
    ${\beta^* \over 2} \ge {\opt(G) \over 2\alpha}$, by Claim~\ref{clm:pass}, for each $\beta$ where $\beta^2 \le \opt(G)$, w.p.
    $1 - o({1 \over \log n})$, $\beta$ passes.  Therefore, we can apply a union bound over all $O(\log \alpha) (= O(\log n))$ choices of
    $\beta$, and claim that w.p. $1 - o(1)$, for all $\beta$ where $\beta^2 \le \opt(G)$, $\beta$ passes. Now, since $\beta^*$ is the
    smallest $\beta$ that fails, we have ${\beta^*}^2 \ge \opt(G)$, which implies
    $\beta^* \ge {\opt(G) \over \beta^*} \ge {\opt(G) \over \alpha}$. Hence, ${\beta^* \over 2} \ge {\opt(G) \over 2\alpha}$.
    
    To see ${\beta^* \over 2} \le \opt(G)$, we consider two cases: $\alpha \le \opt(G)$ or $\alpha > \opt(G)$. If $\alpha \le \opt(G)$, we
    trivially have ${\beta^* \over 2} \le {\alpha \over 2} \le {\opt(G) \over 2} \le \opt(G)$. Now, if $\alpha > \opt(G)$, there exists a
    unique $\beta' \in \set{\log n, 2 \log n, 2^2 \log n, \ldots, \alpha}$ where ${\beta' \over 2} \le \opt(G) \le \beta'$. Then by
    Claim~\ref{clm:fail}, w.h.p. $\beta'$ fails. Since $\beta^*$ is the smallest that fails, $\beta^* \le \beta'$. Hence
    ${\beta^* \over 2} \le {\beta' \over 2} \le \opt(G)$.
  \end{proof}

  It remains to analyze case~$(i)$.
  \begin{lemma}\label{lem:dynamic-all-pass}
    If all $\beta$ passes, $\max\set{\alpha, {{\alpha \over \log^2n} \cdot \tOPT_\alpha}}$ (denoted by $\ALG$) is an
    $O(\alpha)$-approximation of $\opt(G)$.
  \end{lemma}
  \begin{proof}
    Recall that if $\testerg(\galpha, {n \log^2n \over \alpha^2})$ returns \No, $\tOPT_\alpha$ is the estimate returned by $\testerg$, and
    if $\testerg$ returns \Yes (i.e., $\opt(\galpha) \ge \gamma \cdot {n \log^2n \over \alpha^2}$), $\tOPT_\alpha$ is defined to be
    $\gamma \cdot {n \log^2n \over \alpha^2}$.

    Intuitively speaking, $\testerg$ returning \Yes is the special case where the sampled graph $\galpha$ has a matching of size (even)
    larger than ${n \over \alpha^2}$, which implies that $\opt(G)$ itself is very large ($\Omega({n \over \alpha})$ by
    Part~(\ref{part:sample-upper}) of Lemma~\ref{lem:vertex-sample}). In this case, ${n \over \alpha}$ is always an
    $O(\alpha)$-approximation (which is basically $\alpha \cdot \tOPT_\alpha$). We should remark that the expression we use for $\ALG$ is a
    unified expression that works for both \testerg outputs \Yes and \testerg outputs \No.

    We now prove the lemma formally.  First note that for either case, $\tOPT_\alpha \le \opt(\galpha)$.  In the following, we first show
    that $\ALG \le \opt(G)$, and then show that $\ALG \ge {\opt(G) \over O(\alpha)}$.

    To see that $\ALG = \max\set{\alpha, {{\alpha \over \log^2n} \cdot \tOPT_\alpha}}\le \opt(G)$, firstly, if $\alpha > \opt(G)$, then
    there exists $\beta$ used by Algorithm~1 where ${\beta \over 2} \le \opt(G) \le \beta$, and by Claim~\ref{clm:fail}, this $\beta$ fails
    w.p. $1 - o(1)$ (which contradicts to the fact that all $\beta$ pass). Therefore, $\alpha \le \opt(G)$, and we only need to show that
    ${{\alpha \over \log^2n} \cdot \tOPT_\alpha} \le \opt(G)$. As pointed out above, $\opt(G^\alpha) \ge { \tOPT_\alpha}$. Hence, w.h.p.,
  \begin{align*}
   {\alpha \over \log^2n} \cdot \tOPT_\alpha  & \le {\alpha \over \log^2n} \cdot \opt(G^\alpha) \\
& \le {\alpha \over \log^2n} \cdot {3\log n \over \alpha }\cdot \opt(G) \tag{By Lemma~\ref{lem:vertex-sample} Part~(\ref{part:sample-upper})}\\
&\le \opt(G)
  \end{align*}
  proving $\ALG \le \opt(G)$. 

  To see that $\ALG = \max\set{\alpha, {{\alpha \over \log^2n} \cdot \tOPT_\alpha}} = {\opt(G) \over O(\alpha)}$, firstly, if $\opt(G) < \alpha^2$, trivially
  \begin{align*}
    \ALG \ge \alpha >   {\opt(G) \over \alpha}.
  \end{align*}

  Therefore, we only need to consider $\opt(G) \ge \alpha^2$. There are two cases: $\testerg(G^\alpha, {n \log^2n \over \alpha^2})$ returns
  \Yes or returns \No. If $\testerg(G^\alpha, {n \log^2n \over \alpha^2})$ returns \Yes, then
  $\tOPT_\alpha := {\gamma n \log^2 n\over \alpha^2}$, and hence
  \begin{align*}
    \ALG & \ge {\alpha \over \log^2n} \cdot \tOPT_\alpha = {\alpha \over \log^2n} \cdot   {\gamma n \log^2 n \over \alpha^2}\\
         &= {\gamma n \over \alpha} \ge {\gamma \cdot \opt(G) \over \alpha} = {\opt(G) \over O(\alpha)} 
  \end{align*}

  If $\testerg(G^\alpha, {n \log^2n \over \alpha^2})$ returns \No, then by the definition of \testerg, $\tOPT_\alpha \ge \gamma \cdot 
  \opt(\galpha)$. We have, w.h.p.,
  \begin{align*}
    \ALG & \ge {\alpha \over \log^2n} \cdot \tOPT_\alpha \ge {\alpha \over \log^2n} \cdot {\gamma \cdot \opt(G^\alpha)}\\
         &\ge {\gamma \alpha \over \log^2n} \cdot {\log^2 n\over 2\alpha^2}\cdot \opt(G)  
          = {\opt(G) \over O(\alpha)}\tag{Lemma~\ref{lem:vertex-sample} Part~(\ref{part:sample-lower})}
  \end{align*}
  Therefore, $\ALG = {\opt \over O(\alpha)}$ for all $\opt(G) \ge  \alpha$, which completes the proof.
\end{proof}
\end{proof}

\subsection{Implementing  Matching Size Testers in Graph Streams}\label{sec:matching-size-tester}

We now show how to implement matching size testers in insertion-only streams and dynamic streams. 

\begin{claim}\label{clm:tester-insert}
  A $0.5$-matching size tester $\tester_{0.5}(G,k)$ can be implemented in $\Ot(k)$ space in insertion-only streams.
\end{claim}
\begin{proof}
  Simply maintain a maximal matching $M$ and stop when $k/2$ edges have been collected. If $\card{M} = k/2$, return \Yes, and otherwise
  return \No along with $\card{M}$ as the estimate.
\end{proof}

\begin{proof}[Proof of Theorem~\ref{thm:upper-intro}, Part~(1)]
  Suppose Algorithm~1 returns a $c \cdot \alpha$-approximation (Lemma~\ref{lem:alpha-apx}; $c$ is a constant). First notice that if
  $\alpha < c \log n$, $\Ot({n \over \alpha^2})$ space is enough to store a maximal matching of the input graph $G$ which is a
  $2$-approximation of $\opt(G)$.  Therefore, we only need to consider $\alpha \ge c \log n$.  Define $\widehat{\alpha} = \alpha/ c$; we have
  $\widehat{\alpha} \ge \log n$. Run Algorithm~1 for $\widehat{\alpha}$ using the tester by Claim~\ref{clm:tester-insert}.

  By Lemma~\ref{lem:alpha-apx}, Algorithm~1 returns a $c \cdot \widehat{\alpha} (= \alpha)$-approximation of $\opt(G)$ w.h.p.  On the other
  hand, Algorithm~1 invokes $\testerg(*,k)$ for $O(\log \alpha)$ times where the largest $k$ used is
  $\Ot(\max\set{{n \over \alpha^2}, 1}) = \Ot({n \over \alpha^2})$ (recall that $\alpha \le \sqrt{n}$). Therefore, by
  Claim~\ref{clm:tester-insert}, the space requirement is $\Ot({n \over \alpha^2})$.
\end{proof}

For implementing a matching size tester in dynamic streams, we use the following result
from~\cite{AssadiKLY16,ChitnisCEHMMV16}.

\begin{lemma}[\!\!\cite{AssadiKLY16,ChitnisCEHMMV16}]\label{lem:dynamic-size-k}
  There exists a constant $\gamma$ such that a randomized $\gamma$-matching size tester $\tester_{\gamma}(G,k)$ that succeeds w.p. $1-o(\frac{1}{n})$ can be implemented in dynamic streams using $\Ot(k^2)$ space. 
\end{lemma}
One simple approach for implementing a tester for Lemma~\ref{lem:dynamic-size-k} is to randomly group the vertices into $\Theta(k)$ groups and
compute a maximum matching between the groups.  It is shown in~\cite{AssadiKLY16} that this can be done in $\Ot(k^2)$ space, while
w.h.p. the size of the maximum matching between the groups is either $\Omega(k)$ (hence tester outputs \Yes) or $\Omega(\opt)$ (hence tester
outputs \No, along with the matching size).

\begin{proof}[Proof of Theorem~\ref{thm:upper-intro}, Part~(2)]
  Suppose Algorithm~1 returns a $c \cdot \alpha$-approximation (Lemma~\ref{lem:alpha-apx}; $c$ is a constant).  First notice that if
  $\alpha < c \log n$, $\Ot({n^2 \over \alpha^4})$ bits of space is enough to maintain a counter for each edge slot in the input graph $G$,
  which can recover all edges in $G$. Therefore, we only need to consider $\alpha > c \log n$.  Define $\widehat{\alpha} = \alpha/c$; we
  have $\widehat{\alpha} \ge \log n$. Run Algorithm~1 for $\widehat{\alpha}$ using the \testerg by Lemma~\ref{lem:dynamic-size-k}. Since
  Algorithm~1 only invokes \testerg for $O(\log n)$ times, by Lemma~\ref{lem:dynamic-size-k}, w.h.p., no \testerg fails.

  Now by Lemma~\ref{lem:alpha-apx}, Algorithm~1 outputs an $c \cdot \widehat{\alpha} (= \alpha)$-approximation of $\opt(G)$.  On the other
  hand, Algorithm~1 invokes $\testerg(*,k)$ for $O(\log \alpha)$ times where the largest $k$ used is
  $\Ot(\max\set{{n \over \alpha^2}, 1}) = \Ot({n \over \alpha^2})$ (recall that $\alpha \le \sqrt{n}$). Therefore, by
  Lemma~\ref{lem:dynamic-size-k}, the space requirement is $\Ot({n^2 \over \alpha^4})$.
\end{proof}

\subsection*{Acknowledgements}
We thank Michael Kapralov for many helpful discussions. We are also thankful to anonymous reviewers of SODA for many valuable comments.
\clearpage
\bibliographystyle{acm}
\bibliography{general}

\clearpage
\appendix
\newcommand{\neighbor}[1]{\ensuremath{N(#1)}}

\section{An $O(\sqrt{n})$-Approximation Algorithm in $\polylog{(n)}$-Space}\label{app:folklore}

For completeness, we sketch the proof of a simple $O(\sqrt{n})$-approximation algorithm for matching size estimation in dynamic streams. We emphasize here that this 
algorithm is already known in the literature (see, e.g.,~\cite{KapralovKS14}) and is provided here for the sake of completeness. 

\begin{theorem}[Folklore]\label{thm:folklore}
	There exists a $\polylog{(n)}$-space algorithm that with high probability, outputs an $O(\sqrt{n})$-approximation of the maximum matching size in dynamic graph streams. 
\end{theorem}

\begin{proof}[Proof Sketch]
Recall that $\opt:=\opt(G)$ denotes the cardinality of a maximum matching in the graph $G(L,R,E)$. To achieve an $O(\sqrt{n})$-approximation to $\OPT$, we will establish a simple
connection between the cardinality of a maximum matching and the number of neighbors of a set of $\sqrt{n}$ random vertices chosen from $L$.  Let
$S \subseteq L$ be a random set of size $\sqrt{n}$. We denote by $\neighbor{S}$ the set of neighbors of $S$ in the final graph (at the end
of the stream), and let $k = \min\set{\card{\neighbor{S}}, \sqrt{n}}$.  We show that w.h.p
$\Omega(k) \le \OPT \le O(k \sqrt{n})$. Using the $\ell_0$-estimation algorithm of Kane~\etal~\cite{kaneNW2010}, we can estimate $\card{\neighbor{S}}$ to within a
constant factor using $\polylog{(n)}$ space with success probability $.99$ in dynamic steams.  This, together with the aforementioned result, suffices to achieve an
$O(\sqrt{n})$-approximation of the matching size.  Note that the error can be made one-sided in a straightforward way, and the overall probability
of success can be boosted to $(1-1/n)$ by running $O(\log{n})$ parallel copies and taking the median value.

We now briefly explain how the aforementioned relation between $k$ and $\OPT$ is established.  We show that there exist two constants
$c_1, c_2 > 1$ such that for any $k \in [\sqrt{n}]$, if $\OPT <k/c_1$, $\card{\neighbor{S}}$ is less than $k$, and if
$\OPT > c_2 k \sqrt{n}$, $\card{\neighbor{S}}$ is larger than $k$, each with probability at least $0.99$.  If $\OPT < k/c_1$, then by the
extended Hall's Theorem, there exists a set $S'$ of $(n - k/c_1)$ vertices in $L$ that has at most $k/c_1 (< k)$ neighbors. Since
$k\le \sqrt{n}$, the size of $S'$ is large enough to ensure that with a constant probability (which can be made to $0.99$ by choosing a
large enough constant $c_1$), the $\sqrt{n}$ chosen vertices in $S$ are a subset of $S'$, and hence $\card{N(S)} < k$.

On the other hand, if $\OPT > c_2 \cdot k \sqrt{n}$, then there exists two subset of vertices $A \subseteq L$ and $B\subseteq R$, with $\card{A} = \card{B} \geq c_2 \cdot k \sqrt{n}$ such that there exists a perfect matching
between $A$ and $B$ in $G$. Consequently, with a constant probability (which can be made to $0.99$ by choosing a
large enough constant $c_2$), $S$ contains more than $k$ vertices of $A$, and the perfect matching ensures that the number of neighbors of $S$ is more than $k$.
\end{proof}

\end{document}